\documentclass{aart}

\usepackage{graphicx, color}
\usepackage{amsmath} 
\usepackage{amssymb}
\usepackage{amsthm}
\usepackage{amsxtra}
\usepackage{amscd}
\usepackage{bbm}
\usepackage{mathrsfs}
\usepackage{bm}

\renewcommand{\cal}{\mathcal} 
 
\newcommand{\fra}{\mathfrak} 
\newcommand{\ul}[1]{\underline{#1} \!\,} 
\newcommand{\ol}[1]{\overline{#1} \!\,} 

\newcommand{\ee}{\mathrm{e}}
\newcommand{\ii}{\mathrm{i}}
\newcommand{\dd}{\mathrm{d}}

\newcommand{\deq}{\mathrel{\mathop:}=}

\newcommand{\umat}{\mathbbmss{1}} 
\renewcommand{\leq}{\leqslant}
\renewcommand{\geq}{\geqslant}

\renewcommand{\epsilon}{\varepsilon}

\newcommand{\R}{\mathbb{R}}
\newcommand{\C}{\mathbb{C}}
\newcommand{\N}{\mathbb{N}}
\newcommand{\Z}{\mathbb{Z}}

\newcommand{\p}[1]{({#1})}
\newcommand{\pb}[1]{\bigl({#1}\bigr)}
\newcommand{\pB}[1]{\Bigl({#1}\Bigr)}
\newcommand{\pbb}[1]{\biggl({#1}\biggr)}
\newcommand{\pBB}[1]{\Biggl({#1}\Biggr)}

\newcommand{\q}[1]{[{#1}]}
\newcommand{\qb}[1]{\bigl[{#1}\bigr]}
\newcommand{\qB}[1]{\Bigl[{#1}\Bigr]}
\newcommand{\qbb}[1]{\biggl[{#1}\biggr]}
\newcommand{\qBB}[1]{\Biggl[{#1}\Biggr]}

\newcommand{\h}[1]{\{{#1}\}}
\newcommand{\hb}[1]{\bigl\{{#1}\bigr\}}
\newcommand{\hB}[1]{\Bigl\{{#1}\Bigr\}}

\newcommand{\abs}[1]{\lvert #1 \rvert}
\newcommand{\absb}[1]{\bigl\lvert #1 \bigr\rvert}

\newcommand{\absbb}[1]{\biggl\lvert #1 \biggr\rvert}

\newcommand{\norm}[1]{\lVert #1 \rVert}
\newcommand{\normb}[1]{\bigl\lVert #1 \bigr\rVert}
\newcommand{\normB}[1]{\Bigl\lVert #1 \Bigr\rVert}
\newcommand{\normbb}[1]{\biggl\lVert #1 \biggr\rVert}

\newcommand{\scalar}[2]{\langle{#1} \mspace{2mu}, {#2}\rangle}
\newcommand{\scalarb}[2]{\bigl\langle{#1} \mspace{2mu}, {#2}\bigr\rangle}
\newcommand{\scalarB}[2]{\Bigl\langle{#1} \,\mspace{2mu},\, {#2}\Bigr\rangle}

\newcommand{\scalarBB}[2]{\Biggl\langle{#1} \,\mspace{2mu},\, {#2}\Biggr\rangle}

\newcommand{\com}[2]{[{#1} \mspace{2mu}, {#2}]}
\newcommand{\comb}[2]{\bigl[{#1} \mspace{2mu}, {#2}\bigr]}

\newcommand{\bra}[1]{\langle #1 |}
\newcommand{\brab}[1]{\bigl\langle #1 \bigr|}

\newcommand{\ket}[1]{| #1 \rangle}
\newcommand{\ketb}[1]{\bigl| #1 \bigr\rangle}

\DeclareMathOperator{\tr}{Tr}

\DeclareMathOperator{\re}{Re}
\DeclareMathOperator{\im}{Im}

\DeclareMathOperator{\sgn}{sgn}

\theoremstyle{plain} 
\newtheorem{theorem}{Theorem}[section]
\newtheorem*{theorem*}{Theorem}
\newtheorem{lemma}[theorem]{Lemma}
\newtheorem*{lemma*}{Lemma}

\newtheorem*{corollary*}{Corollary}

\newtheorem*{proposition*}{Proposition}

\newtheorem*{definition*}{Definition}

\newtheorem*{example*}{Example}
\newtheorem{remark}[theorem]{Remark}

\newtheorem*{remark*}{Remark}
\newtheorem*{remarks*}{Remarks}

\usepackage{cite}

\usepackage{graphicx}
\usepackage[nottoc,notlof,notlot]{tocbibind}

\usepackage{psfrag}

\usepackage{inline_sub}
\usepackage{verbatim}
\numberwithin{equation}{section}

\newcommand{\A}{\mathbb{A}}
\newcommand{\B}{\mathbb{B}}
\renewcommand{\p}[1]{\left(#1 \right)}
\renewcommand{\q}[1]{\left[#1 \right]}
\renewcommand{\h}[1]{\left\{#1 \right\}}

\begin{document}
\title{A Microscopic Derivation of the Time-Dependent Hartree-Fock Equation with Coulomb Two-Body Interaction}

\author{
J\"urg Fr\"ohlich${}^1$ \qquad Antti Knowles${^2}$
\\
\\
Institute of Theoretical Physics, 
ETH H\"onggerberg,
\\
CH-8093 Z\"urich, Switzerland ${}^1$
\\
\\
Department of Mathematics, Harvard University\\
Cambridge MA 02138, USA${}^1$
}
\date{August 10, 2011}
\maketitle

\begin{abstract}
We study the dynamics of a Fermi gas with a Coulomb interaction potential, and show that, in a mean-field regime, the 
dynamics is described by the Hartree-Fock equation. This extends previous work of Bardos et al.\ \cite{BardosGolse} to 
the case of unbounded interaction potentials. We also express the mean-field limit as a ``superhamiltonian'' system, and 
state our main result in terms of the Heisenberg-picture dynamics of observables. This is a Egorov-type theorem.
\end{abstract}

\section{Introduction} \label{section: introduction}
The Hartree-Fock equation is a fundamental tool, used throughout physics and chemistry, for describing a system consisting of a large number of fermions. Despite its importance for both conceptual and numerical applications, many questions surrounding it remain unsolved.  One area in which significant progress has been made is the microscopic justification of the static Hartree-Fock equation, which is known to yield the correct asymptotic ground state energy of large atoms and molecules; see \cite{LiebSimon1, LiebSimon2, Bach, Fefferman1, Fefferman2, GrafSolovej}. The time-dependent Hartree-Fock equation, which is supposed to describe the dynamics of a large Fermi system, has received less attention. To our knowledge, the only work in which this equation is derived from microscopic Hamiltonian dynamics is \cite{BardosGolse}. The Cauchy problem for the time-dependent Hartree-Fock equation has also been studied in the literature; see \cite{Bove, Chadam} and especially \cite{Zagatti}, where the Cauchy problem is solved for singular interaction potentials.

A key assumption in \cite{BardosGolse} is that the interaction potential be bounded. A goal of this article is to extend the result of \cite{BardosGolse} to a class of singular interaction potentials, which includes the physically relevant Coulomb potential. We also describe how this mean-field result can be formulated as a Egorov-type theorem.

A system of $N$ spinless\footnote{For simplicity of exposition we omit the spin, whose inclusion is merely a notational complication.} fermions is described by a wave function $\Psi_N(x_1, \dots, x_N) \in \bigwedge^N L^2(\R^3, \dd x)$ which is totally antisymmetric in its arguments. The dynamics of $\Psi_N$ is governed by the usual Schr\"odinger equation. In order to obtain a mean-field limit, the Schr\"odinger equation is rescaled with $N$. In this article we adopt the scaling of \cite{BardosGolse}. The Schr\"odinger equation reads
\begin{equation} \label{schroedinger equation}
\ii \partial_t \Psi_N(t) \;=\; H_N \Psi_N(t)\,,
\end{equation}
where the $N$-particle Hamiltonian $H_N$ is defined by
\begin{equation} \label{mean-field Hamiltonian}
H_N \;\deq\; \sum_{i = 1}^N h_i + \frac{1}{N} \sum_{i < j} w(x_i - x_j)\,.
\end{equation}
Here, $h_i$ is a one-particle Hamiltonian acting on the coordinate $x_i$, typically of the form $h_i = -\Delta_i + v(x_i)$, where $\Delta$ is the three-dimensional Laplacian and $v$ is some external potential; $w$ is the interaction potential. Under the assumptions on $v$ and $w$ we make below, it is easy to see that $H_N$ is a well-defined self-adjoint operator with domain $\bigwedge^N H^2(\R^3)$.

We briefly sketch our main result. Consider a sequence of $N$ orthonormal orbitals $\varphi_1, \dots, \varphi_N$, where $\varphi_i$ is a one-particle wave function. This defines an $N$-particle fermionic state through the Slater determinant
\begin{equation*}
\Psi_N \;\deq\; \varphi_1 \wedge \cdots \wedge \varphi_N\,.
\end{equation*}
Let $\Psi_N(t)$ be the solution of the Schr\"odinger equation \eqref{schroedinger equation} with initial state $\Psi_N$.
In general, $\Psi_N(t)$ is no longer a Slater determinant for $t \neq 0$.
However, one expects that this holds asymptotically for large $N$:
\begin{equation*}
\Psi_N(t) \;\approx\; \varphi_1(t) \wedge \cdots \wedge \varphi_N(t)\,.
\end{equation*}
Here the orbitals $\varphi_1(t), \dots, \varphi_N(t)$ are supposed to solve the Hartree-Fock equation
\begin{equation} \label{rescaled Hartree-Fock}
\ii \partial_t \varphi_i \;=\; h \varphi_i + \frac{1}{N} \sum_{j  = 1}^N (w * \abs{\varphi_j}^2) \varphi_i - \frac{1}{N} \sum_{j = 1}^N \p{w * (\varphi_i \bar{\varphi}_j)} \varphi_j\,.
\end{equation}
Our main result (Theorems \ref{theorem: result 1} and \ref{theorem: result 2} below) is a precise formulation of this asymptotic behaviour.

Aside from the mathematical question of generalizing the result of \cite{BardosGolse} to singular potentials,
this result is of some physical relevance when studying the dynamics of electrons of a large atom in the 
Born-Oppenheimer approximation. Consider an atom of atomic number $N$ (commonly also denoted by $Z$). The nucleus has 
charge $Ne$, where $e$ is the positive unit charge, and is surrounded by $N$ electrons of charge $-e$. We assume that 
the nucleus is immobile; this is heuristically justified by the fact that the nucleus is much heavier than the 
electrons. The Hamiltonian of the electrons reads, in appropriately chosen units,
\begin{equation} \label{unrescaled atomic Hamiltonian}
\sum_{i = 1}^N \pBB{-\Delta_i - \frac{e^2 N}{\abs{x_i}}} + \sum_{1 \leq i < j \leq N} \frac{e^2}{\abs{x_i - x_j}}\,.
\end{equation}
After conjugation with the unitary dilation defined by $x_i \mapsto N^{-1} x_i$ for all $i$, the Hamiltonian \eqref{unrescaled atomic Hamiltonian} becomes
\begin{equation*}
N^2 \qBB{\sum_{i = 1}^N \pBB{-\Delta_i - \frac{e^2}{\abs{x_i}}} + \frac{1}{N} \sum_{1 \leq i < j \leq N} \frac{e^2}{\abs{x_i - x_j}}}\,,
\end{equation*}
which is of the form $N^2 H_N$. We conclude that our results describe the dynamics of atomic electrons at length scales of order $N^{-1}$ and time scales of order $N^{-2}$. The approximation is therefore quite crude. For instance in the Thomas-Fermi atom, most electrons are to be found at length scales of order $N^{-1/3}$, while the innermost electrons (K-shell, etc.) reside at length scales of order $N^{-1}$.

One problem in the above physical model, as well as in the works \cite{LiebSimon1, LiebSimon2, Bach, Fefferman1, Fefferman2, GrafSolovej}, is that, as $N$ becomes large, relativistic effects should be taken into account. Indeed, a simple argument shows that the speeds of the innermost electrons are of order $N$. Another problem in applying the time-dependent Hartree-Fock theory to the dynamics of excited states is that the interaction with the radiation field is neglected. This interaction is responsible for the relaxation of excited states to the ground state of the atom.

A somewhat different physical scenario is an interacting Fermi gas confined to a box of fixed size. As discussed in \cite{NarnhoferSewell, ElgartErdosSchleinYau}, the natural scaling in this situation may be viewed as a combination of mean-field and semiclassical scalings. This problem was first studied in \cite{NarnhoferSewell, Spohn}. The authors show that the limiting dynamics is governed by the Vlasov equation. These results were sharpened in \cite{ElgartErdosSchleinYau}, where the authors compare the Hamiltonian dynamics with the dynamics of the Hartree equation, and derive estimates on the rate of convergence.

A further, physically very different, scenario studied in the literature is an interacting Fermi gas in the weak coupling regime. Here the limiting dynamics is given by a nonlinear Boltzmann equation. See for instance \cite{ErdosSalmhoferYau2004}, in which a nonrigorous derivation is given for a model of interacting fermions on a lattice.

Finally, we outline the key ideas of our proof. It relies on the diagrammatic Schwinger-Dyson expansion and Kato smoothing estimates developed in \cite{bosons}. The main steps are:
\begin{itemize}
\item[(a)]
Use the Schwinger-Dyson expansion to express the Hamiltonian time evolution of a $p$-particle observable. 
\item[(b)]
Show that, in the limit $N \to \infty$, only the tree terms of the Schwinger-Dyson expansion survive.
\item[(c)]
Show that the time evolution of a $p$-particle observable under the Hartree-Fock equation converges to the tree terms of the Schwinger-Dyson series as $N \to \infty$.
\end{itemize}
Steps (a) and (b) have been addressed in \cite{bosons}. Thus, the argument in this paper consists in doing step (c).

This article is organized as follows. In Section \ref{section: Hartree-Fock equation} we introduce the Hartree-Fock equation, discuss its Hamiltonian structure and prove a Schwinger-Dyson series for its time evolution. In Section \ref{section: density matrix Hartree-Fock} we rewrite the Hartree-Fock equation using density matrices. In Sections \ref{section: second quantization} and \ref{section: Slater determinants} we introduce second quantized notation and Slater determinants. After these preparations, we state our main result in Section \ref{section: main result}. The proof is given in Section \ref{section: convergence}. Finally, Section \ref{section: Egorov} is devoted to a Egorov-type formulation of our main result, whereby the microscopic dynamics is recognized as a quantization of a classical ``superhamiltonian'' theory.

\subsubsection*{Conventions}
In the following, the expression ``$A(t)$ holds for small times'' is understood to mean that there is a constant $T$ such that $A(t)$ is true for all $\abs{t} < T$. The precise value of $T$ can always be inferred from the context. To simplify notation, we assume in the following that $t \geq 0$.

The norm of a Hilbert space $\mathcal{H}$ is denoted by $\norm{\cdot}$. We denote by
\begin{equation*}
\mathcal{H}^{(n)}_\pm \;\deq\; P_\pm \mathcal{H}^{\otimes n}
\end{equation*}
the symmetric/antisymmetric subspaces of the tensor product space $\mathcal{H}^{\otimes n}$. Here, $P_\pm$ is the orthogonal projector onto the symmetric/antisymmetric subspace. The Banach space of bounded operators on $\mathcal{H}$ with operator norm is denoted by $(\cal L(\mathcal{H}), \norm{\cdot})$, and the Banach space of trace-class operators on $\mathcal{H}$ with trace norm is denoted by $(\mathcal{L}^1(\mathcal{H}), \norm{\cdot}_1)$.

We use the notation $a^{(p)}_{i_1 \dots i_p} \in \cal L(\mathcal{H}^{\otimes n})$ to denote a $p$-particle operator 
$a^{(p)} \in \cal L(\mathcal{H}^{\otimes p})$ acting on the coordinates $x_{i_1}, \dots, x_{i_p} $ of $n$-particle 
space. Similarly, $\tr_{i_1 \dots i_p}$ denotes a partial trace over the $p$-particle space corresponding to the 
coordinates $x_{i_1}, \dots, x_{i_p}$.

A time subscript of the form $(\cdot)_t$ is always understood to mean time evolution up to time $t$ of $(\cdot)$ with respect to the appropriate free dynamics. We shall explain this in greater detail whenever this notation is used.

The symbol $C$ is reserved for a constant whose dependence on some parameters may be indicated. The value of $C$ need not be the same from one equation to the next.

\subsubsection*{Acknowledgements}
We would like to thank two referees for pointing out inaccuracies in an earlier version of this manuscript.

\section{The Hartree-Fock equation} \label{section: Hartree-Fock equation}
For simplicity of notation, we only consider spinless fermions in the following; the one-particle Hilbert space is $\mathcal{H} \deq L^2(\R^3, \dd x) \equiv L^2(\R^3)$. Merely cosmetic modifications extend our results to the case of spin-$s$ fermions for which the one-particle Hilbert space is $L^2(\R^3) \otimes \C^{2s + 1}$. To fix ideas, we consider the free Hamiltonian $h \deq - \Delta$ and a Coulomb two-body interaction potential $w(x) \deq \kappa \abs{x}^{-1}$. By a simple extension of the results of \cite{bosons}, Section 8, our results remain valid for a free Hamiltonian of the form $h = -\Delta + v$ and a two-body interaction potential $w$, where $w$ is even and $v,w \in L^\infty(\R^3) + L^3_w(\R^3)$ are both real. Here, $L^p_w$ denotes the weak $L^p$-space (see e.g.\ \cite{ReedSimonII}). In particular, we may treat Hamiltonians of the form 
\begin{equation*}
\sum_{i = 1}^N \pBB{-\Delta_i - \frac{e^2}{\abs{x_i}}} + \frac{1}{N} \sum_{1 \leq i < j \leq N} \frac{e^2}{\abs{x_i - x_j}}\,,
\end{equation*}
describing the dynamics of electrons in a large atom, as discussed in Section \ref{section: introduction}.

\subsection{Some notation}
It is convenient to state the time-dependent Hartree-Fock equation in terms of an infinite sequence of orbitals $\Phi = (\varphi_i)_{i \in \N}$ which is an element of the Hilbert space
\begin{equation*}
\tilde{\mathcal{H}} \;\deq\; l^2(\N; L^2(\R^3)) \;=\; l^2(\N) \otimes L^2(\R^3)\,.
\end{equation*}
To simplify notation, we set $\alpha = (x,i)$ and write $\Phi(\alpha) = \varphi_i(x)$. Furthermore, we abbreviate
\begin{equation*}
\int \dd \alpha \;\deq\; \sum_{i \in \N} \int \dd x\,, \qquad
\delta(\alpha - \alpha') \;\deq\; \delta_{i i'} \delta(x - x')\,.
\end{equation*}
The scalar product on $\tilde{\mathcal{H}}$ is then given by
\begin{equation*}
\scalar{\Phi}{\Phi'} \;=\; \int \dd \alpha \; \ol{\Phi(\alpha)} \Phi'(\alpha)\,.
\end{equation*}
Let $a^{(p)} \in \cal L(\mathcal{H}^{\otimes p})$ and define $\tilde{a}^{(p)} \in \cal L(\tilde{\mathcal{H}}^{\otimes p})$ through
\begin{equation*}
\tilde{a}^{(p)} \;\deq\; \umat_{(l^2(\N))^{\otimes p}} \otimes a^{(p)}\,.
\end{equation*}
We have the identity 
\begin{equation} \label{norm of tilde operator}
\norm{\tilde{a}^{(p)}} \;=\; \norm{a^{(p)}}\,.
\end{equation}
Furthermore, one easily finds that
\begin{equation} \label{extended operator calculated}
\scalarb{\Phi^{\otimes p}}{\tilde{a}^{(p)} \Phi^{\otimes p}} \;=\; \sum_{i_1, \dots, i_p \in \N} \scalarb{\varphi_{i_1} \otimes \cdots \otimes \varphi_{i_p}}{a^{(p)} \, \varphi_{i_1} \otimes \cdots \otimes \varphi_{i_p}}\,.
\end{equation}

\subsection{Hamiltonian formulation of the Hartree-Fock equation}
The time-dependent Hartree-Fock equation for the sequence $\Phi$ reads
\begin{equation} \label{Hartree-Fock}
\ii \partial_t \varphi_i \;=\; h \varphi_i + \sum_{j \in \N} (w * \abs{\varphi_j}^2) \varphi_i - \sum_{j \in \N} \p{w * (\varphi_i \bar{\varphi}_j)} \varphi_j\,.
\end{equation}
We begin by noting that \eqref{Hartree-Fock} is the Hamiltonian equation of motion of a classical Hamiltonian system with phase space $\Gamma \deq l^2(\N) \otimes H^1(\R^3)$.

Define the map $\A$, from closed operators $A^{(p)}$ on  $\tilde{\mathcal{H}}_+^{(p)}$ to ``polynomial'' functions on phase space, through
\begin{multline*}
\A(A^{(p)})(\Phi) \;\deq\; \;\scalarb{\Phi^{\otimes p}}{A^{(p)} \Phi^{\otimes p}}\,
\\
=\; \int \dd \alpha_1 \cdots \dd \alpha_p \, \dd \beta_1 \cdots \dd \beta_p \; \ol{\Phi(\alpha_p)} \cdots \ol{\Phi(\alpha_1)} A^{(p)}(\alpha_1, \dots, \alpha_p; \beta_1, \dots, \beta_p) \, \Phi(\beta_1) \cdots \Phi(\beta_p)\,,
\end{multline*}
where $A^{(p)}(\alpha_1, \dots, \alpha_p; \beta_1, \dots, \beta_p)$ is the distribution kernel of $A^{(p)}$ (see \cite{bosons} for details). We denote by $\mathfrak{A}$ the linear hull of functions of the form $\A(A^{(p)})$, with $A^{(p)} \in \cal L(\tilde{\mathcal{H}}_+^{(p)})$.

The Hamilton function is given by
\begin{equation} \label{orbital Hamiltonian}
H \;\deq\; \A(\tilde{h}) + \frac{1}{2} \A(\tilde{\mathcal{W}})\,,
\end{equation}
where
\begin{equation*}
\mathcal{W} \;\deq\; W(\umat - E)\,;
\end{equation*}
here $(E \Psi)(x_1,x_2) \deq \Psi(x_2,x_1)$ is the exchange operator and $W$ is the two-particle operator defined by multiplication by $w(x_1 - x_2)$.
Written out in terms of components, \eqref{orbital Hamiltonian} reads
\begin{equation*}
H(\Phi) \;=\; \sum_{i \in \N} \scalar{\varphi_i}{h \varphi_i} + \frac{1}{2} \sum_{i,j \in \N}  \pb{\scalar{\varphi_i \otimes \varphi_j}{W \, \varphi_i \otimes \varphi_j} - \scalar{\varphi_i \otimes \varphi_j}{W \, \varphi_j \otimes \varphi_i}}\,.
\end{equation*}
Using Sobolev-type inequalities, one readily sees that $H$ is well-defined on $\Gamma$.

A short calculation shows that the Hartree-Fock equation is equivalent to
\begin{equation*}
\ii \partial_t \Phi \;=\; \partial_{\bar{\Phi}} H(\Phi)\,.
\end{equation*}
The symplectic form on $\Gamma$ is given by
\begin{equation*}
\omega \;=\; \ii \int \dd \alpha \; \dd \ol{\Phi}(\alpha) \wedge \dd \Phi(\alpha)\,,
\end{equation*}
which induces the Poisson bracket
\begin{equation} \label{Poisson bracket for Hartree-Fock}
\{\Phi(\alpha), \ol{\Phi}(\beta)\} \;=\; \ii \delta(\alpha - \beta) \,, \qquad \{\Phi(\alpha), \Phi(\beta)\} \;=\; \{\ol{\Phi}(\alpha), \ol{\Phi}(\beta)\} \;=\; 0 \,.
\end{equation}
Thus, for two observables $A, B \in \mathfrak{A}$,
\begin{equation*}
\h{A, B}(\Phi) \;=\; \ii \int \dd \alpha \; \p{\frac{\delta A}{\delta \Phi(\alpha)}(\Phi) \frac{\delta B}{\delta \ol{\Phi}(\alpha)}(\Phi) - \frac{\delta B}{\delta \Phi(\alpha)}(\Phi) \frac{\delta A}{\delta \ol{\Phi}(\alpha)}(\Phi)}\,.
\end{equation*}
The Hamiltonian equation of motion on $\Gamma$ is the Hartree-Fock equation \eqref{Hartree-Fock}.

The conservation laws of the Hartree-Fock flow can be understood in terms of symmetries of the Hamiltonian \eqref{orbital Hamiltonian}. One immediately sees that \eqref{orbital Hamiltonian} is invariant under the rotation $\Phi \mapsto (U \otimes \umat_{L^2(\R^3)}) \Phi$, where $U \in \cal L(l^2(\N))$ is unitary. A one-parameter group of such unitary transformations is generated by linear combinations of the functions $\re \scalar{\varphi_i}{\varphi_j}$ and $\im \scalar{\varphi_i}{\varphi_j}$, which Poisson-commute with the Hamiltonian $\eqref{orbital Hamiltonian}$. By Noether's principle, it follows that $\scalar{\varphi_i}{\varphi_j}$ is (at least formally) conserved. The energy $H$ is of course formally conserved as well.

In order to solve the Hartree-Fock equation \eqref{Hartree-Fock} with initial state $\Phi$, we rewrite it as an integral equation
\begin{equation} \label{integral Hartree-Fock}
\varphi_i(t) \;=\; \ee^{-\ii t h} \varphi_i - \ii \int_0^t \dd s \; \ee^{-\ii (t - s) h}\sum_{j \in \N} \pb{(w * \abs{\varphi_j(s)}^2) \varphi_i(s) - \p{w * (\varphi_i(s) \bar{\varphi}_j(s))} \varphi_j(s)}\,.
\end{equation}
The Cauchy-problem for \eqref{integral Hartree-Fock} was solved in \cite{Zagatti}. We quote the relevant results:
\begin{lemma} \label{lemma: Zagatti}
Let $\Phi \in \tilde{\mathcal{H}}$. Then \eqref{integral Hartree-Fock} has a unique global solution $\Phi(\cdot) \in C(\R; \tilde{\mathcal{H}})$. Furthermore, the quantities $\scalar{\varphi_i}{\varphi_j}$ are conserved. In particular, $\norm{\Phi(t)} = \norm{\Phi}$.
\end{lemma}

\subsection{A Schwinger-Dyson expansion for the Hartree-Fock equation}
Our main tool is the Schwinger-Dyson expansion for the flow of the Hartree-Fock equation. We 
use the notation $(\cdot)_t$ to denote free time evolution generated by the free Hamiltonian 
$\A(\tilde{h})$. Explicitly,
\begin{equation*}
A_t(\varphi_1, \varphi_2, \dots) \;=\; A \pb{\ee^{-\ii th } \varphi_1, \ee^{-\ii th } \varphi_2, \dots}\,.
\end{equation*}

\begin{lemma} \label{classical Schwinger-Dyson expansion}
Let $A \in \mathfrak{A}$, $\nu > 0$, and $\Phi(t)$ be the solution of \eqref{integral Hartree-Fock} with initial data $\Phi \in \tilde{\cal H}$. Then, for small times $t$, we have that
\begin{align*}
A(\Phi(t)) &\;=\; A_t(\Phi) + \int_0^t \dd s \; \frac{1}{2} \hb{\A(\tilde{\mathcal{W}}), A_{t-s}}(\Phi(s)) 
\\
&\;=\; \sum_{k = 0}^\infty \frac{1}{2^k} \int_{\Delta^k(t)} \dd \ul{t} \; \hB{\A(\tilde{\mathcal{W}}_{t_k}), \dots \hb{\A(\tilde{\mathcal{W}}_{t_1}), A_t}}(\Phi)\,,
\end{align*}
uniformly for $\Phi \in B_\nu \deq \{\Phi \in \tilde{\mathcal{H}} \,:\, \norm{\Phi}^2 \leq \nu\}$\,.
\end{lemma}

\begin{proof}
The proof of Lemma 7.1 in \cite{bosons} applies with virtually no modifications. One uses \eqref{norm of tilde operator}, the identity
\begin{equation*}
\A(\tilde{\mathcal{W}})_t \;=\; \A(\tilde{\mathcal{W}}_t) \;=\; \A\pb{(W_t (\umat - E))\,\tilde{\;}\,}\,,
\end{equation*}
and $\norm{E} = 1$.
\end{proof}

\section{The density matrix Hartree-Fock equation} \label{section: density matrix Hartree-Fock}
From now on, we only work with orthogonal sequence of orbitals belonging to the set
\begin{equation*}
\mathcal{K} \;\deq\; \{\Phi \in \tilde{\mathcal{H}} \,:\, \scalar{\varphi_i}{\varphi_j} = 0 \text{ for } i\neq j\}\,.
\end{equation*}
By Lemma \ref{lemma: Zagatti}, $\Phi \in \mathcal{K}$ implies that $\Phi(t) \in \mathcal{K}$ for all $t$. To each sequence of orbitals $\Phi$ we assign a one-particle density matrix
\begin{equation} \label{density matrix from sequnce}
\gamma_\Phi \;\deq\; \sum_{i \in \N} \ket{\varphi_i} \bra{\varphi_i}\,.
\end{equation}
It is easy to see that this defines a mapping from $\mathcal{K}$ onto the set of density matrices
\begin{equation*}
\mathcal{D} \;\deq\; \h{\gamma \in \mathcal{L}^1(\mathcal{H}) \,:\, \gamma \geq 0}\,. 
\end{equation*}
Furthermore,
\begin{equation*}
\norm{\gamma_\Phi}_1 \;=\; \norm{\Phi}^2\,.
\end{equation*}
Conversely, one may recover $\Phi$ from $\gamma_\Phi$, up to ordering of the orbitals, by spectral decomposition. 
Also, \eqref{extended operator calculated} implies that
\begin{equation} \label{observables computed using trace}
\A(\tilde{a}^{(p)})(\Phi) \;=\; \tr (a^{(p)} \gamma_\Phi^{\otimes p})\,.
\end{equation}

Next, we note that the Hartree-Fock equation may be formulated in terms of density matrices.
Let $\Phi(t)$ be a solution of the Hartree-Fock equation \eqref{Hartree-Fock}, and abbreviate
\begin{equation*}
\gamma(t) \;=\; \gamma_{\Phi(t)}\,.
\end{equation*}
Then a short calculation shows that
\begin{equation} \label{Hartree-Fock for density matrices}
\ii \partial_t \gamma \;=\; [h, \gamma] + \tr_2 \q{\mathcal{W}, \gamma \otimes \gamma}\,,
\end{equation}
which is the Hartree-Fock equation for density matrices.
As an integral equation in the interaction picture, this reads
\begin{equation} \label{integral Hartree-Fock for density matrices}
\gamma(t) \;=\; \ee^{-\ii th} \, \gamma \, \ee^{\ii th} - \ii \int_0^t \dd s \; \ee^{- \ii (t-s) h} \, \tr_2 \q{\mathcal{W}, \gamma(s) \otimes \gamma(s)} \, \ee^{\ii (t-s) h} \,.
\end{equation}
Sometimes it is convenient to rewrite this using the shorthand
\begin{equation} \label{interaction picture density matrix}
\tilde{\gamma}(t) \;\deq\; \ee^{\ii th} \,\gamma(t) \,\ee^{-\ii th}\,.
\end{equation}
Then \eqref{integral Hartree-Fock for density matrices} is equivalent to
\begin{equation} \label{integral Hartree-Fock for density matrices interaction}
\tilde{\gamma}(t) \;=\; \gamma - \ii \int_0^t \dd s \; \tr_2 \q{\mathcal{W}_s, \tilde{\gamma}(s) \otimes \tilde{\gamma}(s)}\,.
\end{equation}
The next lemma ensures that if $\Phi(t)$ is a general solution of the integral Hartree-Fock equation \eqref{integral Hartree-Fock} then $\gamma_{\Phi(t)}$ solves the integral density matrix equation \eqref{integral Hartree-Fock for density matrices}.
\begin{lemma} \label{Hartree-Fock solves density matrix Hartree-Fock}
Let $\Phi(t)$ be the solution of \eqref{integral Hartree-Fock}. Then $\gamma_{\Phi(t)}$ solves \eqref{integral Hartree-Fock for density matrices}.
\end{lemma}
\begin{proof}
Let $a^{(1)} \equiv a \in \cal L(\mathcal{H})$. From 
Lemma \ref{classical Schwinger-Dyson expansion} we get
\begin{equation} \label{integral equation for one-particle observable}
\A(\tilde{a})(\Phi(t)) \;=\; \A(\tilde{a}_t)(\Phi) + \int_0^t \dd s \; \hb{\A(\tilde{\mathcal{W}}), \A(\tilde{a}_{t-s})}(\Phi(s))\,. 
\end{equation}
Now \eqref{Poisson bracket for Hartree-Fock} and \eqref{observables computed using trace} imply
\begin{align*}
\hb{\A(\tilde{\mathcal{W}}), \A(\tilde{a})}(\Phi) &\;=\; \ii \A\pb{\qb{\tilde{\mathcal{W}} , \tilde{a} \otimes \umat}}(\Phi) 
\\
&\;=\; \ii \tr \pb{\qb{\mathcal{W} , a \otimes \umat} \gamma_\Phi \otimes \gamma_\Phi}
\\
&\;=\; - \ii \tr \pb{(a \otimes \umat) \qb{\mathcal{W} , \gamma_\Phi \otimes \gamma_\Phi}}\,.
\end{align*}
Thus \eqref{integral equation for one-particle observable} reads
\begin{align*}
\tr (a \, \gamma_{\Phi(t)}) &\;=\; \tr (a_t \, \gamma_\Phi) - \ii \int_0^t \dd s \; \tr \pb{(a_{t - s} \otimes \umat) \qb{\mathcal{W} , \gamma_{\Phi(s)} \otimes \gamma_{\Phi(s)}}}
\\
&\;=\; \tr \pb{a \, \ee^{\ii t h} \gamma_\Phi \ee^{- \ii th}} - \ii \int_0^t \dd s \; \tr \pB{a \, \ee^{- \ii (t-s) h} \, \tr_2 \q{\mathcal{W}, \gamma_{\Phi(s)} \otimes \gamma_{\Phi(s)}} \, \ee^{\ii (t-s) h} }\,.
\end{align*}
Since $a \in \cal L(\mathcal{H})$ was arbitrary, this is equivalent to \eqref{integral Hartree-Fock for density matrices}.
\end{proof}

\section{Second quantization} \label{section: second quantization}
For the following it is convenient to use second quantized notation; see e.g.\ \cite{BratteliRobinson} and \cite{bosons} for a full account. We introduce the fermionic Fock space
\begin{equation*}
\cal F \;\deq\; \bigoplus_{N \geq 0} \cal H^{(N)}_-\,,
\end{equation*}
where we adopt the usual convention that $\cal H^{(0)}_- = \C$.
A vector $\Psi \in \cal F$ is a sequence $\Psi = (\Psi_0, \Psi_1, \Psi_2, \dots)$ with $\Psi_N \in \cal H^{(N)}_-$ for all $N$. By a slight abuse of notation, we often identify an $N$-particle vector $\Psi_N \in \cal H^{(N)}_-$ with the vector in Fock space whose $N$-particle component equals $\Psi_N$ and whose other components vanish.

On $\cal F$ act the usual fermionic creation and annihilation operators, $a^*$ and $a$, which map the one-particle space into densely defined closable operators on $\cal F$. For $\varphi \in \cal H$ and $\Psi \in \cal F$, they are defined by
\begin{align*}
\pb{a^*(\varphi) \Phi}_N(x_1, \dots, x_N) &\;\deq\; \frac{1}{\sqrt{N}} \sum_{i = 1}^N (-1)^{i+1} \varphi(x_i) \Psi_{N-1}(x_1, \dots, x_{i - i}, x_{i+1}, \dots, x_N)\,,
\\
\pb{a(\varphi) \Phi}_N(x_1, \dots, x_N) &\;\deq\; \sqrt{N+1} \int \dd y \; \ol{\varphi(y)} \Psi_{N+1}(y, x_1, \dots, x_N)\,.
\end{align*}
It is not hard to see that $a^*(\varphi)$ and $a(\varphi)$ are each other's adjoints. Moreover, they satisfy the canonical anticommutation relations
\begin{equation*}
\comb{a(\varphi)}{a^*(\varphi')}_+ \;=\; \scalar{\varphi}{\varphi'}\,, \qquad \comb{a^\sharp(\varphi)}{a^\sharp(\varphi')}_+ \;=\; 0\,,
\end{equation*}
where $\com{A}{B}_+ \deq AB + BA$ is the anticommutator and $a^\sharp$ stands for either $a^*$ or $a$.

Define the operator-valued distributions $a^*(x) \deq a^*(\delta_x)$ and $a(x) \deq a(\delta_x)$, where $\delta_x$ is Dirac's delta function centred at $x$. In other words,
\begin{equation*}
a^*(\varphi) \;=\; \int \dd x \; \varphi(x) \, a^*(x)\,, \qquad a(\varphi) \;=\; \int \dd x \; \ol{\varphi(x)} \, a(x)\,.
\end{equation*}
To streamline notation, it is convenient to introduce the rescaled creation and annihilation operators, defined by
\begin{equation*}
a_\nu^\sharp(x) \;\deq\; \frac{1}{\sqrt{\nu}} \, a^\sharp(x)\,.
\end{equation*}
Here $\nu > 0$ is a parameter that will ultimately be taken to equal $N$, the number of particles.

Let $a^{(p)} \in \cal L(\cal H^{(p)})$ and define its second quantization $\widehat{\A}_\nu(a^{(p)})$, a closed operator on $\cal F$, through
\begin{equation*}
\widehat{\A}_\nu(a^{(p)}) \;\deq\; \int \dd x_1 \cdots \dd x_p \, \dd y_1 \cdots \dd y_p \, a^*_\nu(x_p) \cdots a^*_\nu(x_1)\, a^{(p)}(x_1, \dots, x_p; y_1, \dots, y_p) \, a_\nu(y_1) \cdots a_\nu(y_p)\,,
\end{equation*}
where $a^{(p)}(x_1, \dots, x_p; y_1, \dots, y_p)$ denotes the distribution kernel of $a^{(p)}$. Explicitly, $\widehat{\A}_\nu(a^{(p)})$ is given by (see \cite{bosons})
\begin{equation} \label{second quantization on n-particle space}
\widehat{\A}_\nu(a^{(p)}) \Bigr|_{\cal H^{(N)}_-} \;=\;
\begin{cases}
\frac{p!}{\nu^p} \binom{N}{p} P_- (a^{(p)} \otimes \umat^{(N-p)}) P_- &\text{if } N\geq p
\\
0 &\text{if } N<p\,,
\end{cases}
\end{equation}
which may be viewed as an alternative definition of $\widehat{\A}_\nu(a^{(p)})$. In particular, when restricted to $\cal H^{(N)}_-$, the operator $\widehat{\A}_N(a^{(p)})$ is of order one as $N \to \infty$.

\section{Slater determinants} \label{section: Slater determinants}

Next, we introduce quasi-free states (see \cite{BratteliRobinson} for more details). Their importance for our purposes stems from the fact that the Hartree-Fock equation naturally describes the time evolution of quasi-free states. Let $\gamma \in \cal D$ be a one-particle density matrix. The \emph{quasi-free state} $\omega_\gamma$ associated with $\gamma$ satisfies by definition
\begin{equation*}
\gamma^{(p)}(x_1, \dots, x_p; y_1, \dots, y_p) \;=\; \det \p{\gamma(x_i; y_j)}_{i,j}\,,
\end{equation*}
where
\begin{equation*}
\gamma^{(p)}(x_1, \dots, x_p; y_1, \dots, y_p) \;\deq\; \omega_\gamma\pb{a^*(y_p) \cdots a^*(y_1) a(x_1) \cdots a(x_p)}
\end{equation*}
is the reduced $p$-particle density matrix of $\omega_\gamma$.
In other words, $\gamma^{(p)}$ is the operator kernel of
\begin{equation} \label{p-particle action of quasi-free state}
\gamma^{(p)} \;=\; \gamma^{\otimes p} \, \Sigma_-^{(p)} \,,
\end{equation}
where
\begin{equation*}
\Sigma_-^{(p)} \;\deq\; p! P_-^{(p)}\,.
\end{equation*}
For the following calculations it is convenient to introduce the symbol $\epsilon_{i_1 \dots i_p}^{j_1 \dots j_p}$, which is equal to $\sgn \sigma$ if $i_1, \dots, i_p$ are disjoint and there is a permutation $\sigma \in S_p$ such that $(i_1, \dots, i_p) = (j_{\sigma(1)}, \dots, j_{\sigma(p)})$, and equal to $0$ otherwise. Also, for the remainder of this section, summation over any index appearing twice in an equation is implied.
\begin{lemma} \label{lemma: bound on gamma_p}
Let $\gamma \in \mathcal{D}$ with $\tr \gamma = 1$. Then $\tr \gamma^{(p)} \leq 1$.
\end{lemma}
\begin{proof}
There is an orthonormal basis $(\varphi_i)_{i \in \N}$ and a sequence of nonnegative numbers $(\lambda_i)_{i \in \N}$ such that $\sum_i \lambda_i = 1$ and $\gamma = \sum_i \lambda_i \ket{\varphi_i} \bra{\varphi_i}$. Therefore,
\begin{equation*}
\gamma^{(p)} \;=\; \epsilon_{i_1 \dots i_p}^{j_1 \dots j_p} \lambda_{i_1} \cdots \lambda_{i_p} \ket{\varphi_{i_1} \otimes \cdots \otimes \varphi_{i_p}} \bra{\varphi_{j_1} \otimes \cdots \otimes \varphi_{j_p}}
\end{equation*}
This yields
\begin{align*}
\tr \gamma^{(p)} &\;=\; \epsilon_{i_1 \dots i_p}^{j_1 \dots j_p} \lambda_{i_1} \cdots \lambda_{i_p} \delta_{i_1 k_1} \cdots \delta_{i_p k_p} \delta_{j_1 k_1} \cdots \delta_{j_p k_p}
\\
&\;=\; \sum_{i_1, \dots, i_p \text{ disjoint}} \lambda_{i_1} \cdots \lambda_{i_p}
\\
&\;\leq\; \sum_{i_1, \dots, i_p} \lambda_{i_1} \cdots \lambda_{i_p} \;=\; 1\,.
\qedhere
\end{align*}
\end{proof}

Next, we introduce a special class of quasi-free states, described by \emph{Slater determinants}.
Take an orthonormal sequence of orbitals $\Phi = (\varphi_i)_{i \in \N}$ and denote by $\Phi^{(N)}$ the truncated sequence $(\varphi_1, \dots, \varphi_N, 0, \dots)$.
We define the $N$-particle Slater determinant as
\begin{equation*}
S(\Phi^{(N)}) \;\deq\; \varphi_1 \wedge \cdots \wedge \varphi_N \;=\; \frac{1}{\sqrt{N!} }a^*(\varphi_N) \cdots a^*(\varphi_1) \Omega \;=\; \sqrt{N!} \, P^{(N)}_- \, \varphi_1 \otimes \cdots \otimes \varphi_N \;\in\; \mathcal{H}^{(N)}_-\,.
\end{equation*}
Note that the normalization is chosen so that $\norm{S(\Phi^{(N)})} = 1$. The corresponding $N$-particle density matrix is
\begin{equation*}
\Gamma_N \;\deq\; \ket{S(\Phi^{(N)})} \bra{S(\Phi^{(N)})}\,.
\end{equation*}
The $p$-particle marginals of $\Gamma_{N}$ are given by
\begin{align}
\Gamma_N^{(p)} &\;\deq\; \tr_{p+1 \dots N} \Gamma_N
\notag \\
&\;=\; \tr_{p+1 \dots N} \frac{1}{N!} \, \epsilon_{i_1 \dots i_N} \epsilon_{j_1 \dots j_N} \, \ketb{\varphi_{i_1} \otimes \cdots \otimes \varphi_{i_N}} \brab{\varphi_{j_1} \otimes \cdots \otimes \varphi_{j_N}}
\notag \\ \label{p-particle marginals computed}
&\;=\; \frac{(N-p)!}{N!} \,\epsilon_{i_1 \dots i_p}^{j_1 \dots j_p} \, \ketb{\varphi_{i_1} \otimes \cdots \otimes \varphi_{i_p}} \brab{\varphi_{j_1} \otimes \cdots \otimes \varphi_{j_p}}\,.
\end{align}

In particular,
\begin{equation} \label{one-particle density matrix of a Slater determinant}
\Gamma_N^{(1)} \;=\; \frac{1}{N} \sum_{i = 1}^N \ket{\varphi_i} \bra{\varphi_i} \;=\; \gamma_N\,,
\end{equation}
where $\gamma_N$ is the one-particle density matrix associated (by \eqref{density matrix from sequnce}) with the normalized truncated sequence
\begin{equation*}
\tilde{\Phi}^{(N)} \;\deq\; \frac{1}{\sqrt{N}} \Phi^{(N)}\,.
\end{equation*}
Thus, the explicit forms \eqref{p-particle action of quasi-free state} and \eqref{p-particle marginals computed} imply a relation between the reduced density matrices and the marginals:
\begin{equation} \label{relation between marginals}
\gamma_N^{(p)} \;=\; \frac{1}{N^p} \, \epsilon_{i_1 \dots i_p}^{j_1 \dots j_p} \, \ketb{\varphi_{i_1} \otimes \cdots \otimes \varphi_{i_p}} \brab{\varphi_{j_1} \otimes \cdots \otimes \varphi_{j_p}} \;=\; \frac{p!}{N^p} \binom{N}{p} \, \Gamma_N^{(p)}\,.
\end{equation}
In other words, Slater determinants determine quasi-free states by their $p$-particle marginals. The normalization $\frac{p!}{N^p} \binom{N}{p}$ differs slightly from the usual normalization $1$ of quasi-free states, but in the limit $N \to \infty$ this difference vanishes.
Recalling \eqref{second quantization on n-particle space}, we see from \eqref{relation between marginals} that
\begin{equation} \label{expectation in Slater determinant}
\scalarb{S(\Phi^{(N)})}{\widehat{\A}_N(a^{(p)}) S(\Phi^{(N)})} \;=\; \tr \pb{a^{(p)} \gamma_N^{(p)}}\,.
\end{equation}

We also note that
\begin{equation} \label{operator norm of gamma}
\norm{\gamma_N} \;=\; \frac{1}{N}\,.
\end{equation}
This is a special case of the well-known statement (see e.g.\ \cite{LiebLoss}) that $\normb{\tr_{2 \dots N} \Gamma} \leq N^{-1}$, for any fermionic $N$-particle density matrix $\Gamma$ \,\footnote{This can also be inferred from \eqref{operator norm of gamma} by writing $\Gamma$ as a linear combination of projectors.}. The estimate \eqref{operator norm of gamma} will play a fundamental role in our analysis.

Finally, we remark that the sequence $\Phi^{(N)}(t)$ satisfies the rescaled Hartree-Fock equation \eqref{rescaled Hartree-Fock} if and only if the normalized sequence $\tilde{\Phi}^{(N)}(t)$ satisfies the unrescaled Hartree-Fock equation \eqref{Hartree-Fock}. Similarly,
$\tilde{\Phi}^{(N)}(t)$ is a solution of the integral equation \eqref{integral Hartree-Fock} if and only if $\Phi^{(N)}(t)$ is a solution of the rescaled integral equation
\begin{equation} \label{rescaled integral Hartree-Fock}
\varphi_i(t) \;=\; \ee^{-\ii t h} \varphi_i - \frac{\ii}{N} \int_0^t \dd s \; \ee^{-\ii (t - s) h} \sum_{j =1}^N \pb{(w * \abs{\varphi_j(s)}^2) \varphi_i(s) - \p{w * (\varphi_i(s) \bar{\varphi}_j(s))} \varphi_j(s)}\,.
\end{equation}

\section{The limit and main result} \label{section: main result}

We may now state our main result. Take an infinite sequence $\Phi = (\varphi_i)_{i \in \N}$ of orthonormal orbitals, and denote the truncated sequences by $\Phi^{(N)}$. Let $\Phi^{(N)}(t)$ be the solution of the rescaled integral Hartree-Fock equation \eqref{rescaled integral Hartree-Fock} with initial data $\Phi^{(N)}$.

\begin{theorem} \label{theorem: result 1}
Let $p \in \N$ and $a^{(p)} \in \cal L (\cal H^{(p)}_-)$. Then, for any $t \in \R$, we have
\begin{equation*}
\scalarB{\ee^{-\ii t H_N} S(\Phi^{(N)})}{\widehat{\A}_N(a^{(p)}) \, \ee^{-\ii t H_N} S(\Phi^{(N)})} - \scalarB{S\pb{\Phi^{(N)}(t)}}{\widehat{\A}_N(a^{(p)}) S\pb{\Phi^{(N)}(t)}} \;\longrightarrow\; 0
\end{equation*}
as $N \to \infty$.
\end{theorem}

We may also express our main result in terms of density matrices. Denote by
\begin{equation*}
\Gamma_N(t) \;\deq\; \ee^{-\ii t H_N} \, \ketb{S(\Phi^{(N)})} \brab{S(\Phi^{(N)})} \, \ee^{\ii t H_N}
\end{equation*}
the $N$-particle density matrix evolved in time using the $N$-body dynamics. Similarly, denote by
\begin{equation*}
\tilde{\Gamma}_N(t) \;\deq\; \ketb{S\pb{\Phi^{(N)}(t)}} \brab{S\pb{\Phi^{(N)}(t)}}
\end{equation*}
the $N$-particle density matrix evolved in time using the Hartree-Fock dynamics. Denote by $\Gamma^{(p)}_N(t)$ and $\tilde{\Gamma}_N^{(p)}(t)$ their respective $p$-particle marginals.
\begin{theorem} \label{theorem: result 2}
Let $p \in \N$. Then, for any $t \in \R$, we have
\begin{equation*}
\lim_{N \to \infty} \normb{\Gamma_N^{(p)}(t) - \tilde{\Gamma}_N^{(p)}(t)}_1 \;=\; 0\,.
\end{equation*}
\end{theorem}

\begin{remark}
The limit $N \to \infty$ of $\Gamma^{(p)}_N(t)$ does not exist in $\norm{\cdot}_1$. Indeed, $\lim_{N \to \infty} \norm{\Gamma^{(p)}_N(t)} = 0$ but $\tr \Gamma^{(p)}_N(t) = 1$ (similarly for $\tilde{\Gamma}_N^{(p)}(t)$).
\end{remark}

\begin{remark}
As mentioned in the beginning of Section \ref{section: Hartree-Fock equation}, both Theorems \ref{theorem: result 1} and \ref{theorem: result 2} extend trivially to the case of spin-$s$ fermions. In that case, we replace the space $\tilde {\cal H}$ from Section \ref{section: Hartree-Fock equation} with the space $l^2(\N) \otimes L^2(\R^3; \C^{2s + 1})$. Thus, $\alpha = (x,\sigma,i)$ where $\sigma = -s, -s+1, \dots, s$ denotes the spin index. Vectors $\Phi = (\varphi_i) \in \tilde{\cal H}$ are now sequences of wave functions $\varphi_i \equiv \varphi_i(x,\sigma)$ which also depend on the spin index $\sigma$. The exchange operator $E$ acts in the natural way: $(E \Psi)(x_1, \sigma_1; x_2, \sigma_2) = \Psi(x_2, \sigma_2; x_1, \sigma_1)$. With these minor modifications, the statements and proofs of Theorems \ref{theorem: result 1} and \ref{theorem: result 2} may be taken over verbatim.
\end{remark}

\begin{remark}
As in \cite{bosons}, one can show that the rate of convergence in Theorems \ref{theorem: result 1} and \ref{theorem: result 2} is a power law $N^{-\beta(t)}$, with $\beta(t) > 0$ for all $t$. However, $\beta(t) \to 0$ as $t \to \infty$. Our bound on the rate of convergence is therefore far from the expected optimal rate $\beta(t) = 1$, which we only obtain for short times.
\end{remark}

\section{Proof of Theorems \ref{theorem: result 1} and \ref{theorem: result 2}} \label{section: convergence}
The main tool of our proof is the graph expansion scheme developed in \cite{bosons}.

\subsection{The Schwinger-Dyson graph expansion} \label{subsection: Schwinger-Dyson}
For the convenience of the reader we summarize the relevant results of the graph expansion in \cite{bosons}. For details and proofs we refer to \cite{bosons}.

Let $a^{(p)} \in \cal L(\mathcal{H}^{(p)}_-)$. From now on, we restrict all operators on $\cal F$ to $\cal H^{(N)}_-$; in particular, we understand expressions of the form $\widehat{\A}_N(a^{(p)})$ to mean $\widehat{\A}_N(a^{(p)}) \bigr|_{\cal H^{(N)}_-}$. We start with the Schwinger-Dyson series for the time-evolved operator $\ee^{\ii t H_N} \widehat{\A}_N(a^{(p)}) \ee^{-\ii t H_N}$. We write
\begin{equation*}
H_N \;=\; N\qbb{\widehat{\A}_N(h) + \frac{1}{2} \widehat{\A}_N(W)}\,,
\end{equation*}
and regard the second term as a perturbation. Thus we get the series expansion
\begin{equation} \label{multiple commutator expansion}
\ee^{\ii t H_N} \, \widehat{\A}_N(a^{(p)}) \, \ee^{-\ii t H_N} \;=\; \sum_{k = 0}^\infty \int_{\Delta^k(t)} \dd \ul{t} \; \frac{(\ii N)^k}{2^k} \qB{\widehat{\A}_N(W_{t_k}), \dots \qB{\widehat{\A}_N(W_{t_1}), \widehat{\A}_N(a^{(p)}_t)}\dots}\,,
\end{equation}
where $\ul{t} = (t_1, \dots, t_k)$ and $\Delta^k(t)$ is the $k$-simplex $\{(t_1, \dots, t_k) : 0 < t_k < \dots < t_1 < t\}$. Here, as before, a time subscript refers to free time evolution:
\begin{equation*}
a^{(p)}_t \;\deq\; \ee^{\ii \sum_{i} h_i t} \, a^{(p)} \, \ee^{-\ii \sum_i h_i t}\,.
\end{equation*}
In \cite{bosons} it was shown how the normal ordering of the multiple commutators in \eqref{multiple commutator expansion} gives rise to terms that can be classified graphically. The graph expansion reads
\begin{equation} \label{loop expansion}
\ee^{\ii t H_N} \, \widehat{\A}_N (a^{(p)}) \, \ee^{- \ii t H_N} \;=\; \sum_{k = 0}^\infty \sum_{l = 0}^k \frac{1}{N^l} \, \widehat{\A}_N \pb{G^{(k,l)}_t(a^{(p)})}\,.
\end{equation}
The $(p+k-l)$-particle operator $G^{(k,l)}_t(a^{(p)})$ corresponds to the sum of all $l$-loop graphs contributing to the multiple commutator of order $k$. Explicitly, it is given by
\begin{equation}
G^{(k,l)}_{t}(a^{(p)}) \;\deq\; \int_{\Delta^k(t)} \dd \ul{t} \; G^{(k,l)}_{t,\ul{t}}(a^{(p)})\,,
\end{equation}
where the operators $G^{(k,l)}_{t, \ul{t}}(a^{(p)})$ are recursively defined by
\begin{align}
G^{(k,l)}_{t,t_1,\dots,t_k}(a^{(p)}) 
\;=\; \;&\ii P_- \sum_{i = 1}^{p+k-l-1} \qB{W_{i \, p+k-l, t_k}, G^{(k-1,l)}_{t, t_1, \dots, t_{k - 1}}(a^{(p)}) \otimes \umat} P_- 
\notag \\ \label{recursive definition of graphs}
&{}+{} \ii P_-  \sum_{1 \leq i < j \leq p+k-l} \qB{W_{ij, t_k}, G^{(k-1, l-1)}_{t, t_1, \dots, t_{k - 1}}(a^{(p)})} P_- \,,
\end{align}
as well as $G^{(0,0)}_t(a^{(p)}) \deq a^{(p)}_t$. If $l < 0$ or $l > k$ then $G^{(k,l)}_{t,t_1,\dots,t_k}(a^{(p)}) = 0$.

As a sum over graphs\footnote{called graph structures in \cite{bosons}}, $G^{(k,l)}_{t}(a^{(p)})$ reads
\begin{equation} \label{graph expansion}
G^{(k, l)}_t(a^{(p)}) \;=\; \frac{\ii^k}{2^k} \sum_{\mathcal{Q} \in \mathscr{Q}(p,k,l)} i_{\mathcal{Q}} \int_{\Delta^k_{\mathcal{Q}}(t)} \dd \ul{t} \; G^{(k, l)(\mathcal{Q})}_{t,t_1,\dots, t_k}(a^{(p)})\,,
\end{equation}
where $i_{\mathcal{Q}} \in \{0,1\}$, $\Delta^k_{\mathcal{Q}}(t) \subset [0,t]^k$, and
$\mathscr{Q}(p,k,l)$ is a set of graphs whose cardinality satisfies
\begin{equation} \label{bound on number of graphs}
\abs{\mathscr{Q}(p,k,l)} \;\leq\; 2^k \binom{k}{l} \, \binom{2p + 3k}{k} \, (p+k-l)^l\,.
\end{equation}
The operator $G^{(k, l)(\mathcal{Q})}_{t,t_1,\dots, t_k}(a^{(p)})$ is an \emph{elementary term}, indexed by the graph $\mathcal{Q}$, of the form
\begin{equation} \label{elementary term}
P_- \, W_{i_1 j_1, t_{v_1}} \cdots W_{i_r j_r, t_{v_r}} \,\pb{a_t^{(p)} \otimes \umat^{(k-l)}} \, W_{i_{r+1} j_{r+1}, t_{v_{r+1}}} \cdots W_{i_k j_k, t_{v_k}} P_- \,,
\end{equation}
where $r = 0, \dots, k$ and $(t_{v_1}, \dots, t_{v_k})$ is some permutation of $(t_1, \dots, t_k)$.

The operator norm of $G^{(k,l)}_{t}(a^{(p)})$ may now be bounded using the dispersive estimate
\begin{equation} \label{general Kato smoothing}
\int \dd t \; \normb{\abs{x}^{-1} \ee^{\ii t \Delta} \psi}^2 \;\leq\; \pi \norm{\psi}^2\,.
\end{equation}
Going to centre of mass coordinates and using Cauchy-Schwarz, one sees that \eqref{general Kato smoothing} implies
\begin{equation} \label{Kato smoothing}
\int_0^t \dd s \; \normb{W_{i j, s} \, \varphi} \;\leq\; \sqrt{\frac{\pi \kappa^2 t}{2}} \, \norm{\varphi}\,.
\end{equation}
Recalling the estimate \eqref{bound on number of graphs}, it is now easy to argue, as in \cite{bosons}, that \eqref{loop expansion} converges uniformly in $N$ for small times $t$. Moreover, the large-$N$ asymptotics of the Schwinger-Dyson series \eqref{loop expansion} is given by the tree terms: for small times $t$ we have
\begin{equation} \label{one loop expansion}
\ee^{\ii t H_N} \, \widehat{\A}_N (a^{(p)}) \, \ee^{- \ii t H_N} \;=\; \sum_{k = 0}^\infty \widehat{\A}_N \pb{G^{(k,0)}_t(a^{(p)})} + L_N(t)\,,
\end{equation}
where $L_N(t)$, corresponding to the sum of all terms with at least one loop ($l \geq 1$), satisfies the estimate 
\begin{equation} \label{estimate on loop terms}
\norm{L_N(t)} \;\leq\; C(p,\kappa, t) \norm{a^{(p)}} N^{-1}\,,
\end{equation}
for small times $t$.

\subsection{Convergence of the Hartree-Fock time evolution to the tree terms}
We now give the main argument of our proof. We show that the Hartree-Fock time evolution is asymptotically ($N \to \infty$) given by the tree terms (i.e.\ the terms $l = 0$) of the Schwinger-Dyson series \eqref{loop expansion}. This result is summarized in Lemma \ref{lemma HF and tree terms} below.

The main idea of the proof is to iterate the integral Hartree-Fock equation so as to obtain a Schwinger-Dyson-type series. If iterated, the Hartree-Fock evolution of the observable $\widehat{\A}_N(a^{(p)})$ yields a power series expansion that \emph{differs} from the tree expansion
\begin{equation*}
\sum_{k = 0}^\infty \widehat{\A}_N \pb{G^{(k,0)}_t(a^{(p)})}\,.
\end{equation*}
Thus, at each step of the iteration we extract an error term, and continue the iteration on what remains. This is done in such a way that the resulting power series is equal to the tree expansion. The main work is to estimate the error terms arising at each step of the iteration. This is done by using a tree expansion combined with the dispersive estimate \eqref{Kato smoothing}.

We work with the density matrix formulation \eqref{integral Hartree-Fock for density matrices} of the Hartree-Fock equation. We seek an expansion for the quantity
\begin{equation} \label{quantity to be expanded}
\scalarb{S(\Phi^{(N)}(t))}{\widehat{\A}_N(a^{(p)}) S(\Phi^{(N)}(t))} \;=\; \tr \pb{a^{(p)} \gamma_N^{(p)}(t)}\,;
\end{equation}
see \eqref{expectation in Slater determinant}. In this subsection, the special form \eqref{one-particle density matrix of a Slater determinant} of $\gamma_N$ is unimportant. We therefore assume that we have an arbitrary orthogonal sequence $\Phi = (\varphi_i)_{i \in \N} \in \mathcal{K}$, and denote by $\Phi(t)$ the solution of the Hartree-Fock equation \eqref{integral Hartree-Fock} with initial data $\Phi$. Let $\gamma(t) \deq \gamma_{\Phi(t)}$ be the associated one-particle density matrix. 

By choosing $A = \A(\tilde{a}^{(p)})$, $a^{(p)} \in \cal L(\mathcal{H}^{(p)}_-)$, in Lemma \ref{classical Schwinger-Dyson expansion} and mimicking the proof of Lemma \ref{Hartree-Fock solves density matrix Hartree-Fock} one finds that
\begin{equation} \label{Hartree-Fock on tensor power, step 1}
\tr \pb{a^{(p)} \, \gamma(t)^{\otimes p}} \;=\; \tr \pb{a^{(p)}_t \gamma^{\otimes p}} - \ii \int_0^t \dd s \sum_{i = 1}^p \tr \pB{a^{(p)}_{t-s} \tr_{p+1} \qb{\mathcal{W}_{i \; p+1}, \gamma(s)^{\otimes (p+1)}}}\,.
\end{equation}
It is convenient to use the representation $\tilde{\gamma}(t)$ defined in \eqref{interaction picture density matrix}. Using the substitution $a^{(p)} \mapsto a^{(p)}_{-t}$ in \eqref{Hartree-Fock on tensor power, step 1}, we get
\begin{equation*}
\tr \pb{a^{(p)} \, \tilde{\gamma}(t)^{\otimes p}} \;=\; \tr \pb{a^{(p)} \gamma^{\otimes p}} - \ii \int_0^t \dd s \sum_{i = 1}^p \tr \pB{a^{(p)} \tr_{p+1} \qb{\mathcal{W}_{i \; p+1, s}, \tilde{\gamma}(s)^{\otimes (p+1)}}}\,,
\end{equation*}
where, we recall,
\begin{equation*}
\tilde{\gamma}(t) \;\deq\; \ee^{\ii th} \,\gamma(t) \,\ee^{-\ii th}\,.
\end{equation*}
Recall that $\mathcal{W}_{ij} = W_{ij} (\umat - E_{ij})$. Also, $E_{ij}$ commutes with $W_{ij}$ and with $\tilde{\gamma}(s)^{\otimes (p+1)}$. Thus $\Sigma^{(p)}_- a^{(p)} = p! \, a^{(p)}$, together with the fact that $a^{(p)}$ was arbitrary, implies the integral equation
\begin{equation} \label{Hartree-Fock on tensor power}
\tilde{\gamma}(t)^{\otimes p} \, \Sigma_-^{(p)} \;=\; \gamma^{\otimes p} \, \Sigma_-^{(p)} - \ii \int_0^t \dd s \; \tr_{p+1} \q{\sum_{i = 1}^p W_{i \; p+1, s}, \tilde{\gamma}(s)^{\otimes (p+1)} (\umat - E_{i \; p+1})} \Sigma_-^{(p)}.
\end{equation}

The tree expansion would be obtained by iterating the somewhat different integral equation
\begin{equation} \label{Duhamel formula giving tree expansion}
\tilde{\gamma}(t)^{\otimes p} \, \Sigma_-^{(p)} \;=\; \gamma^{\otimes p} \, \Sigma_-^{(p)} - \ii \int_0^t \dd s \; \tr_{p+1} \q{\sum_{i = 1}^p W_{i \; p+1, s}, \tilde{\gamma}(s)^{\otimes (p+1)} \, \Sigma_-^{(p+1)}}.
\end{equation}
In order to compare \eqref{Hartree-Fock on tensor power} with \eqref{Duhamel formula giving tree expansion}, we use the elementary identity
\begin{equation*}
\Sigma^{(p+1)}_- \;=\; \pbb{\umat - \sum_{j = 1}^p E_{j \; p+1}}\Sigma_-^{(p)}\,.
\end{equation*}
Thus we find
\begin{align*}
\tr_{p+1} \q{\sum_{i = 1}^p W_{i \; p+1, s}, \tilde{\gamma}(s)^{\otimes (p+1)} \, \Sigma_-^{(p+1)}}
&\;=\; \tr_{p+1} \q{\sum_{i = 1}^p W_{i \; p+1, s}, \tilde{\gamma}(s)^{\otimes (p+1)} \pbb{\umat - \sum_{j = 1}^p E_{j \; p+1}}\Sigma_-^{(p)}}
\\
&\;=\; \tr_{p+1} \q{\sum_{i = 1}^p W_{i \; p+1, s}, \tilde{\gamma}(s)^{\otimes (p+1)} \pbb{\umat - \sum_{j = 1}^p E_{j \; p+1}}} \Sigma_-^{(p)}\,.
\end{align*}
Together with \eqref{Hartree-Fock on tensor power} this yields
\begin{equation} \label{HF evolution of a product state}
\tilde{\gamma}(t)^{\otimes p} \, \Sigma_-^{(p)} \;=\; \gamma^{\otimes p} \, \Sigma_-^{(p)} - \ii \int_0^t \dd s \; \tr_{p+1} \q{\sum_{i = 1}^p W_{i \; p+1, s}, \tilde{\gamma}(s)^{\otimes (p+1)} \, \Sigma_-^{(p+1)}}
+ R_p(t)\,,
\end{equation}
with an error term
\begin{equation} \label{definition of R}
R_p(t) \;\deq\;  - \ii \sum_{1 \leq i \neq j \leq p} \int_0^t \dd s \; \tr_{p+1} \q{W_{i \; p+1, s}, \tilde{\gamma}(s)^{\otimes (p+1)} E_{j \; p+1}} \Sigma_-^{(p)}.
\end{equation}
The partial trace is most conveniently computed using operator kernels. We find
\begin{multline*}
\p{W_{i \; p+1, s} \tilde{\gamma}(s)^{\otimes (p+1)} E_{j \; p+1}} (x_1, \dots, x_{p+1}; y_1, \dots, y_{p+1})
\\
=\; \int \dd z_1 \, \dd z_2 \; \q{\prod_{r \neq i,j} \tilde{\gamma}(s)(x_r; y_r)}
\, W_s(x_i, x_{p+1}; z_1, z_2) \, \tilde{\gamma}(s)(z_1; y_i) \, \tilde{\gamma}(s)(x_j; y_{p+1}) \, \tilde{\gamma}(s)(z_2; y_j)\,,
\end{multline*}
so that
\begin{align}
&\qquad \tr_{p+1} \p{W_{i \; p+1, s} \tilde{\gamma}(s)^{\otimes (p+1)} E_{j \; p+1}} (x_1, \dots, x_p; y_1, \dots, y_p)
\notag \\
&=\; \int \dd z_1 \, \dd z_2 \, \dd z_3 \; \q{\prod_{r \neq i,j} \tilde{\gamma}(s)(x_r; y_r)}
\, W_s(x_i, z_3; z_1, z_2) \, \tilde{\gamma}(s)(z_1; y_i) \, \tilde{\gamma}(s)(x_j; z_3) \, \tilde{\gamma}(s)(z_2; y_j)
\notag \\ \label{partial trace computation}
&=\; \p{\tilde{\gamma}_j(s) W_{i j,s} \tilde{\gamma}(s)^{\otimes p}} (x_1, \dots, x_p; y_1, \dots, y_p)\,.
\end{align}
The second term of the commutator in \eqref{definition of R} is the adjoint of the first and we get 
\begin{align}
R_p(t) &\;=\;  - \ii \sum_{1 \leq i \neq j \leq p} \int_0^t \dd s \; \pB{\tilde{\gamma}_j(s) W_{i j,s} \tilde{\gamma}(s)^{\otimes p} - \tilde{\gamma}(s)^{\otimes p} W_{i j,s} \tilde{\gamma}_j(s)} \Sigma^{(p)}_-
\notag \\ \notag
&\;=\;  - \ii \sum_{1 \leq i \neq j \leq p} \int_0^t \dd s \; \pB{\tilde{\gamma}_j(s) W_{i j,s} \tilde{\gamma}(s)^{\otimes p} \Sigma^{(p)}_- - \tilde{\gamma}(s)^{\otimes p} \Sigma^{(p)}_- W_{i j,s} \tilde{\gamma}_j(s)}\,.
\end{align}

We remark at this point that the formula \eqref{partial trace computation} is the key identity behind our proof. From it our strategy is already apparent: the factor $\tilde \gamma_j(s)$ multiplies a trace class operator, so that the trace norm of the right-hand side of \eqref{partial trace computation} may eventually (after the smoothing effect of the free time evolution has been exploited) be estimated using the \emph{operator norm} of $\tilde \gamma_j(s)$, which is bounded by $N^{-1}$. Moreover, it is also apparent why it is crucial to keep the exchange term in \eqref{Hartree-Fock on tensor power} for our argument to work. Indeed, if the exchange term in \eqref{Hartree-Fock on tensor power} were dropped, the restriction $i \neq j$ in \eqref{definition of R} would no longer hold. In that case we would have to estimate the term
\begin{equation*}
\tr_{p+1} \q{W_{i \; p+1, s}, \tilde{\gamma}(s)^{\otimes (p+1)} E_{i \; p+1}}\,,
\end{equation*}
for which the computation of \eqref{partial trace computation} does not hold.

We are now ready to iterate \eqref{HF evolution of a product state}. Multiplying \eqref{HF evolution of a product state} by $a^{(p)} \in \cal L(\mathcal{H}^{(p)}_-)$ yields
\begin{multline*}
\tr \p{a_t^{(p)} \, \tilde{\gamma}(t)^{\otimes p} \, \Sigma^{(p)}_-}
\;=\; \tr \p{a_t^{(p)} \, \gamma^{\otimes p} \, \Sigma^{(p)}_-} 
\\
{}+{} \ii \int_0^t \dd s \, \sum_{i = 1}^p \tr \p{\q{W_{i \; p+1,s}, a^{(p)}_t \otimes \umat} \, \tilde{\gamma}(s)^{\otimes (p+1)} \, \Sigma^{(p+1)}_-} + \tr \p{a^{(p)}_t \, R_p(t)}\,.
\end{multline*}
Iterating this $K$ times yields our main expansion
\begin{align}
\tr \p{a_t^{(p)} \, \tilde{\gamma}(t)^{\otimes p} \, \Sigma_-^{(p)}} &\;=\; \sum_{k = 0}^{K-1} \int_{\Delta^k(t)} \dd \ul{t} \; \tr \p{G_{t, \ul{t}}^{(k,0)}(a^{(p)}) \, \gamma^{\otimes (p+k)} \Sigma_-^{(p+k)}}
\notag \\
&\qquad {}+{} \int_{\Delta^K(t)} \dd \ul{t} \; \tr \p{G_{t, \ul{t}}^{(K,0)}(a^{(p)}) \, \tilde{\gamma}(t_K)^{\otimes (p+K)} \Sigma_-^{(p+K)}}
\notag \\ \label{main expansion}
&\qquad {}+{} \sum_{k = 0}^{K-1} \sum_{1 \leq i \neq j \leq p+k} R^{k}_{ij}(t)\,,
\end{align}
where the error terms are given by
\begin{multline*}
R^{k}_{ij}(t) \;\deq\; - \ii \, \int_{\Delta^{k+1}(t)} \dd \ul{t} \; \tr \Bigl(G^{(k,0)}_{t, t_1, \dots, t_k}(a^{(p)}) \, \tilde{\gamma}_j(t_{k+1}) \, W_{ij, t_{k + 1}} \, \tilde{\gamma}(t_{k+1})^{\otimes (p+k)} \, \Sigma_-^{(p+k)} \\
- W_{ij, t_{k+1}} \, \tilde{\gamma}_j(t_{k+1}) \, G^{(k,0)}_{t, t_1, \dots, t_k}(a^{(p)}) \, \tilde{\gamma}(t_{k+1})^{\otimes (p+k)} \, \Sigma_-^{(p+k)}
\Bigr)\,.
\end{multline*}

Next, we estimate $R^{k}_{ij}(t)$. Let us concentrate on the first term, which we rewrite using the renaming $t_{k+1} \to s$ as
\begin{equation} \label{half of the error term}
\int_{\Delta^k(t)} \dd \ul{t} \int_0^{\wedge \ul{t}} \dd s \; \tr \p{G^{(k,0)}_{t, t_1, \dots, t_k}(a^{(p)}) \, \tilde{\gamma}_j(s) \, W_{ij, s} \, \tilde{\gamma}(s)^{\otimes (p+k)} \, \Sigma_-^{(p+k)}}\,,
\end{equation}
where $\wedge \ul{t} \deq \min\{t_1, \dots, t_k\}$. The outline of our strategy is as follows. We expand both $G^{(k,0)}_{t, t_1, \dots, t_k}(a^{(p)})$ and $\tilde{\gamma}(s)^{\otimes (p+k)}$ using a graph expansion. The integral over $s$ allows us to control the singularity in the potential $W_{ij,s}$ by invoking the dispersive estimate \eqref{Kato smoothing}. The term $\tilde{\gamma}_j(s)$ is bounded by its operator norm.

We start by deriving a tree expansion for $\tilde{\gamma}(s)^{\otimes p}$.
\begin{lemma} \label{lemma: tree expansion for gamma}
Let $a^{(p)} \in \cal L(\mathcal{H}_-^{(p)})$. For small times we have the tree expansion
\begin{equation} \label{tree expansion of p-particle observable}
\tr \pb{a^{(p)} \tilde{\gamma}(t)^{\otimes p} \Sigma_-^{(p)}} \;=\;
\sum_{k = 0}^\infty \int_{\Delta^k(t)} \dd \ul{t} \; \tr\pb{T^{(k)}_{\ul{t}}(a^{(p)}) \gamma^{\otimes (p+k)} \Sigma_-^{(p)}}\,,
\end{equation}
where $T^{(k)}_{\ul{t}}$ is the linear operator defined by $T^{(0)}(a^{(p)}) \deq a^{(p)}$ and
\begin{equation*}
T^{(k)}_{t_1 \dots t_k}(a^{(p)}) \;=\; \ii \sum_{i = 1}^{p+k-1} \q{\mathcal{W}_{i \, p + k, t_k}, T^{(k-1)}_{t_1 \dots t_{k-1}}(a^{(p)}) \otimes \umat}\,.
\end{equation*}
\end{lemma}
\begin{proof}
Lemma \ref{classical Schwinger-Dyson expansion} applied to $A = \A(\tilde{a}^{(p)})$ yields
\begin{equation*}
\tr \pb{a^{(p)} \tilde{\gamma}(t)^{\otimes p}} \;=\;
\sum_{k = 0}^\infty \int_{\Delta^k(t)} \dd \ul{t} \; \tr\pb{T^{(k)}_{\ul{t}}(a^{(p)}) \gamma^{\otimes (p+k)}}\,.
\end{equation*}
The claim then follows by noting that $\Sigma^{(p)}_-a^{(p)} = p! a^{(p)}$ and that $\sum_{i = 1}^{p+k-1} \mathcal{W}_{i\, p+k, t_k}$ commutes with $\Sigma_-^{(p)}$.
The proof of convergence of the series is the same as the proof of convergence of the series \eqref{loop expansion} outlined in Section \ref{subsection: Schwinger-Dyson}.
\end{proof}

We use Lemma \ref{lemma: tree expansion for gamma} to expand $\tilde{\gamma}(s)^{\otimes (p+k)}$ in \eqref{half of the error term}. The result is that \eqref{half of the error term} equals
\begin{equation*}
\sum_{k' = 0}^\infty \int_{\Delta^k(t)} \dd \ul{t} \int_0^{\wedge \ul{t}} \dd s \int_{\Delta^{k'}(s)} \dd \ul{t}' \;
\tr\hB{T_{\ul{t}'}^{(k')} \pB{G_{t,\ul{t}}^{(k,0)}(a^{(p)}) \tilde{\gamma}_j(s) W_{ij, s}} \gamma^{\otimes (p+k+k')} \Sigma_-^{(p+k)}}\,.
\end{equation*}

Next, we recall from \eqref{graph expansion} that $G^{(k,0)}_{t, t_1, \dots, t_k}(a^{(p)})$ can be written as a sum over tree graphs $\mathcal{Q} \in \mathscr{Q}(p,k,0)$ of elementary terms of the form \eqref{elementary term}. Also, since the definition of $T^{(k)}_{t_1, \dots, t_k}(a^{(p)})$ is the same as the definition of $G^{(k,0)}_{0, t_1, \dots, t_k}(a^{(p)})$ with $W$ replaced by $\mathcal{W}$, we immediately get that $T^{(k)}_{t_1, \dots, t_k}(a^{(p)})$ is equal to a sum over tree graphs $\mathcal{Q} \in \mathscr{Q}(p,k,0)$ of elementary terms of the form
\begin{equation*}
P_- \, \mathcal{W}_{i_1 j_1, t_{v_1}} \cdots \mathcal{W}_{i_r j_r, t_{v_r}} \,\pb{a_t^{(p)} \otimes \umat^{(k-l)}} \, \mathcal{W}_{i_{r+1} j_{r+1}, t_{v_{r+1}}} \cdots \mathcal{W}_{i_k j_k, t_{v_k}} P_- \,.
\end{equation*}
In particular, using the dispersive estimate \eqref{Kato smoothing} and the bound \eqref{bound on number of graphs}, one readily sees that the tree expansion \eqref{tree expansion of p-particle observable} converges for small times.

Applying the tree expansion to both $G^{(k,0)}_{t, t_1, \dots, t_k}(a^{(p)})$ and $\tilde{\gamma}(s)^{\otimes (p+k)}$ in \eqref{half of the error term}, we see that \eqref{half of the error term} is equal to 
\begin{multline} \label{fully expanded error term}
\sum_{k' = 0}^\infty \frac{\ii^{k+k'}}{2^{k+k'}}\sum_{\mathcal{Q} \in \mathscr{Q}(p,k,0)} \sum_{\mathcal{Q}' \in \mathscr{Q}(p+k, k',0)} i_{\mathcal{Q}} i_{\mathcal{Q}'} \int_{\Delta^k_{\mathcal{Q}}(t)} \dd \ul{t} \int_0^{\wedge \ul{t}} \dd s \int_{\Delta^{k'}_{\mathcal{Q}'}(s)} \dd \ul{t}'
\\
\tr\hB{A \, a^{(p)}_{1\dots p,t} \, B P_-^{(p+k)} 
\tilde{\gamma}_j(s) W_{ij, s} \, C \, \gamma^{\otimes (p+k+k')} \Sigma_-^{(p+k)}}\,,
\end{multline}
where $A$, $B$, and $C$ are operators with the following properties:
\begin{enumerate}
\item $A$, $B$, and $C$ depend on the variables $(\mathcal{Q}, \mathcal{Q}', k, k', \ul{t}, \ul{t}')$;
\item
$A$, $B$, and $C$ are each a product of operators of the form $W_{i' j', r}$, or $\mathcal{W}_{i' j', r}$, where $r$ stands for a time variable in $\{t_1, \dots, t_k, t_1', \dots, t_{k'}'\}$;
\item
the product $ABC$ contains $k$ factors $W$ and $k'$ factors $\mathcal{W}$;
\item
each time variable in $t_1, \dots, t_k, t_1, \dots, t_{k'}$ appears exactly once in the product $ABC$.
\end{enumerate}

Next, we estimate the operator norm of the operator multiplying $\gamma^{\otimes (p+k+k')}$ in \eqref{fully expanded error term}.
Let $\varphi \in \mathcal{H}^{\otimes (p+k+k')}$ and estimate
\begin{multline*}
I \;\deq\; \normbb{\sum_{\mathcal{Q} \in \mathscr{Q}(p,k,0)} \sum_{\mathcal{Q}' \in \mathscr{Q}(p+k, k',0)} i_{\mathcal{Q}} i_{\mathcal{Q}'} \int_{\Delta^k_{\mathcal{Q}}(t)} \dd \ul{t} \int_0^{\wedge \ul{t}} \dd s \int_{\Delta^{k'}_{\mathcal{Q}'}(s)} \dd \ul{t}' \;
A \, a^{(p)}_{1\dots p,t} \, B P_-^{(p+k)} 
\tilde{\gamma}_j(s) W_{ij, s} \, C \, \varphi}
\\
\leq\; 
\sum_{\mathcal{Q} \in \mathscr{Q}(p,k,0)} \sum_{\mathcal{Q}' \in \mathscr{Q}(p+k, k',0)} \int_{[0,t]^k} \dd \ul{t} \int_0^t \dd s \int_{[0,t]^{k'}} \dd \ul{t}' \;
\normB{A \, a^{(p)}_{1\dots p, t} \, B P_-^{(p+k)} 
\tilde{\gamma}_j(s) W_{ij, s} \, C \, \varphi}
\end{multline*}
We now perform all time integrations, starting from the left, and using at each step the estimate \eqref{Kato smoothing} as well as
\begin{equation*}
\int_0^t \dd r \; \normb{\mathcal{W}_{i' j', r} \varphi} \;\leq\; \sqrt{2 \pi \kappa^2 t} \, \norm{\varphi}\,,
\end{equation*}
which follows trivially from \eqref{Kato smoothing}.
Also, Lemma \ref{lemma: Zagatti} implies that $\norm{\tilde{\gamma}(s)} = \norm{\gamma}$. Thus we find that  
\begin{equation*}
I \;\leq\;
\sum_{\mathcal{Q} \in \mathscr{Q}(p,k,0)} \sum_{\mathcal{Q}' \in \mathscr{Q}(p+k, k',0)} \pbb{\frac{\pi \kappa^2 t}{2}}^{(k+1)/2} \p{2\pi \kappa^2 t}^{k'/2}
\norm{a^{(p)}} \norm{\gamma} \norm{\varphi}
\end{equation*}
Using the bound
\begin{equation*}
\abs{\mathscr{Q}(p,k,0)} \;\leq\; 4^p\,16^k\,,
\end{equation*}
which can be inferred from \eqref{bound on number of graphs}, we find
\begin{align*}
I &\;\leq\; 4^p \, 16^k \, 4^{p+k} 16^{k'} \pbb{\frac{\pi \kappa^2 t}{2}}^{(k+1)/2} \p{2\pi \kappa^2 t}^{k'/2}
\norm{a^{(p)}} \norm{\gamma} \norm{\varphi}
\\
&\;\leq\;
16^p \sqrt{2 \pi \kappa^2 t} \, \pb{32 \sqrt{2 \pi \kappa^2 t}}^k \, \pb{16 \sqrt{2 \pi \kappa^2 t}}^{k'} \, \norm{a^{(p)}} \norm{\gamma} \norm{\varphi}\,.
\end{align*}
Let $t < (2^{11} \pi \kappa^2)^{-1}$. Now Lemma \ref{lemma: bound on gamma_p} implies that $\norm{\gamma^{\otimes (p+k+k')} \Sigma_-^{(p+k)}}_1 \leq 1$. Using the inequality $\tr(A \Gamma) \leq \norm{A} \norm{\Gamma}_1$ we therefore find that \eqref{half of the error term} is bounded by
\begin{equation*}
16^p \sum_{k' = 0}^\infty \pb{32 \sqrt{2 \pi \kappa^2 t}}^k \, \pb{16 \sqrt{2 \pi \kappa^2 t}}^{k'} \, \norm{a^{(p)}} \norm{\gamma} 
\;=\; 16^p \frac{\pb{32 \sqrt{2 \pi \kappa^2 t}}^k}{1 - 16 \sqrt{2 \pi \kappa^2 t}} \, \norm{a^{(p)}} \norm{\gamma}\,.
\end{equation*}

The second term of $R_{ij}^k(t)$ is equal to the complex conjugate of the first. We thus arrive at the desired bound
\begin{equation} \label{bound on R}
\abs{R_{ij}^k(t)} \;\leq\; 2 \cdot 16^p \frac{\pb{32 \sqrt{2 \pi \kappa^2 t}}^k}{1 - 16 \sqrt{2 \pi \kappa^2 t}} \, \norm{a^{(p)}} \norm{\gamma}\,.
\end{equation}
Therefore the last line of \eqref{main expansion} is bounded by
\begin{multline*}
2 \cdot 16^p \frac{1}{1 - 16 \sqrt{2 \pi \kappa^2 t}} \, \norm{a^{(p)}} \norm{\gamma} \sum_{k = 0}^\infty (p+k)^2 \pb{32 \sqrt{2 \pi \kappa^2 t}}^k
\\
\leq\;
4 \cdot 16^p \, \ee^p \frac{1}{1 - 16 \sqrt{2 \pi \kappa^2 t}} \frac{1}{\pb{1 - 32 \sqrt{2 \pi \kappa^2 t}}^3} \, \norm{a^{(p)}} \norm{\gamma}\,,
\end{multline*}
where we used the estimate $\sum_{k = 0}^\infty (p+k)^L \, x^k \leq \frac{\ee^p \, L!}{(1 - x)^{L+1}}$.

Next, we note that the second line of \eqref{main expansion}, i.e.\ the rest term, vanishes in the limit $K \to \infty$. The procedure is almost identical to (in fact easier than) the above estimation of $\abs{R_{ij}^k(t)}$. The result is
\begin{equation*}
\absbb{\int_{\Delta^K(t)} \dd \ul{t} \; \tr \p{G_{t, \ul{t}}^{(K,0)}(a^{(p)}) \, \tilde{\gamma}(t_K)^{\otimes (p+K)} \Sigma_-^{(p+K)}}}
\;\leq\; 2 \cdot 16^p \frac{\pb{32 \sqrt{2 \pi \kappa^2 t}}^K}{1 - 16 \sqrt{2 \pi \kappa^2 t}} \, \norm{a^{(p)}} \;\longrightarrow\; 0\,,
\end{equation*}
as $K \to \infty$.

We summarize what we have proven: For small times the Hartree-Fock time evolution is equal to the tree expansion plus an error term of order $\norm{\gamma}$.
\begin{lemma} \label{lemma HF and tree terms}
Let $a^{(p)} \in \cal L(\mathcal{H}_-^{(p)})$. Then, for small times, we have
\begin{equation*}
\absbb{\tr \pB{a^{(p)} \, \gamma(t)^{\otimes p} \, \Sigma_-^{(p)}} - \sum_{k = 0}^\infty \tr \pB{G_{t}^{(k,0)}(a^{(p)}) \, \gamma^{\otimes (p+k)} \Sigma_-^{(p+k)}}} \;\leq\; \norm{a^{(p)}} \, \norm{\gamma} \, C(p, \kappa,t)\,,
\end{equation*}
for some constant $C(p, \kappa, t)$.
\end{lemma}

\subsection{Conclusion of the proof}
We now have all the necessary ingredients to prove our main result. Take an infinite sequence 
$\Phi = (\varphi_i)_{i \in \N}$ of orthonormal orbitals, and denote by $\Phi^{(N)}(t)$ the 
solution of the rescaled integral Hartree-Fock equation \eqref{rescaled integral Hartree-Fock} 
with initial data $\Phi^{(N)}$.

Then, for small times, we find from \eqref{one loop expansion} and \eqref{quantity to be expanded} that
\begin{align*}
&\scalarB{\ee^{-\ii t H_N} S(\Phi^{(N)})}{\widehat{\A}_N(a^{(p)}) \, \ee^{-\ii t H_N} S(\Phi^{(N)})} - \scalarB{S\pb{\Phi^{(N)}(t)}}{\widehat{\A}_N(a^{(p)}) S\pb{\Phi^{(N)}(t)}}
\\
&\;=\; \sum_{k = 0}^\infty \scalarB{S(\Phi^{(N)})}{\widehat{\A}_N \pb{G^{(k,0)}_t(a^{(p)})} S(\Phi^{(N)})} + \scalarb{S(\Phi^{(N)})}{L_N(t) S(\Phi^{(N)})} - \tr \pb{a^{(p)} \gamma_N^{(p)}(t)}
\\
&\;=\; \sum_{k = 0}^\infty \tr \pB{G_{t}^{(k,0)}(a^{(p)}) \, \gamma_N^{(p+k)}} - \tr \pb{a^{(p)} \gamma_N^{(p)}(t)} + \scalarb{S(\Phi^{(N)})}{L_N(t) S(\Phi^{(N)})}\,.
\end{align*}
Thus, Lemma \ref{lemma HF and tree terms} and \eqref{estimate on loop terms} imply
\begin{multline} \label{main estimate for small times}
\absbb{\scalarB{\ee^{-\ii t H_N} S(\Phi^{(N)})}{\widehat{\A}_N(a^{(p)}) \, \ee^{-\ii t H_N} S(\Phi^{(N)})} - \scalarB{S\pb{\Phi^{(N)}(t)}}{\widehat{\A}_N(a^{(p)}) S\pb{\Phi^{(N)}(t)}}}
\\
\leq\; C(p, \kappa, t) \norm{a^{(p)}} \pb{\norm{\gamma_N} + N^{-1}} \;\leq\; \frac{2 C(p, \kappa, t)}{N} \norm{a^{(p)}}\,,
\end{multline}
by \eqref{operator norm of gamma}.

Next, we observe that the quantum-mechanical and the Hartree-Fock time-evolutions are norm-preserving. We may therefore iterate \eqref{main estimate for small times}, as in \cite{bosons}, to extend it to all times. Thus we find, for all times $t \in \R$,
\begin{multline} \label{main estimate for all times}
\absbb{\scalarB{\ee^{-\ii t H_N} S(\Phi^{(N)})}{\widehat{\A}_N(a^{(p)}) \, \ee^{-\ii t H_N} S(\Phi^{(N)})} - \scalarB{S\pb{\Phi^{(N)}(t)}}{\widehat{\A}_N(a^{(p)}) S\pb{\Phi^{(N)}(t)}}}
\\
\leq\; C(p, \kappa, t) \norm{a^{(p)}} \pb{\norm{\gamma_N} + N^{-1}} \;\leq\; C(p, \kappa, t) f(t,N) \norm{a^{(p)}}\,,
\end{multline}
where $f(t,N)$ is some function satisfying $\lim_N f(t,N) = 0$ for all $t \in \R$. Theorem \ref{theorem: result 1} follows immediately.

Finally, we prove Theorem \ref{theorem: result 2}. Let $\Gamma_N^{(p)}(t)$ and $\tilde{\Gamma}_N^{(p)}(t)$ be defined as in Section \ref{section: main result}. Thus, plugging \eqref{relation between marginals} and \eqref{expectation in Slater determinant} into \eqref{main estimate for all times} yields
\begin{equation*}
\frac{p!}{N^p} \binom{N}{p} \absbb{\tr \qB{\pB{\Gamma_N^{(p)}(t) - \tilde{\Gamma}_N^{(p)}(t)} a^{(p)}}} \;\leq\; C(p, \kappa, t) f(t,N) \norm{a^{(p)}}\,.
\end{equation*}
Thus, by the duality $(\cal L^1)^* = \cal L$, we find
\begin{equation*}
\normb{\Gamma_N^{(p)}(t) - \tilde{\Gamma}_N^{(p)}(t)}_1 \;\leq\; C(p,\kappa, t) f(t,N)\,.
\end{equation*}
Theorem \ref{theorem: result 2} follows.

\section{A Egorov-type result for small times} \label{section: Egorov}
In this section we describe how the many-body dynamics of fermions may be seen as the quantization of a classical ``superhamiltonian'' system, whose dynamics is approximately described by the Hartree-Fock equation.

\subsection{A graded algebra of observables}
We start by defining a Grassmann algebra of anticommuting variables over the one-particle space $\mathcal{H} = L^2(\R^3)$, and equip it with a suitable norm. Formally, we consider the infinite-dimensional Grassmann algebra generated by $\{\ol{\psi(x)}, \psi(x)\}_{x \in \R^3}$. As it turns out, this algebra can be made into a Banach algebra under a natural choice of norm. This norm is most conveniently formulated by identifying elements of the Grassmann algebra with bounded operators between $L^2$-spaces.

Let
\begin{equation}
a \;=\; (a^{(p,q)})_{p,q \in \N} \,, \qquad a^{(p,q)} \in \cal L(\mathcal{H}^{(q)}_-; \mathcal{H}^{(p)}_-),
\end{equation}
be a family of bounded operators. Such objects will play the role of observables in the following.
By a slight abuse of notation we identify $a^{(p,q)}$ with the family obtained by adjoining zeroes to it.

Define
\begin{equation*}
\mathfrak{B}^G \;\deq\; \{a = (a^{(p,q)}) \,:\, a^{(p,q)} = 0 \text{ for all but finitely many } (p,q)\}\,.
\end{equation*}
We introduce a norm on $\mathfrak{B}^G$ through
\begin{equation}
\norm{a}_{\mathfrak{B}^G} \;\deq\; \sum_{p,q \in \N} \norm{a^{(p,q)}}\,,
\end{equation}
and define $\ol{\mathfrak{B}}^G$ as the completion of $\mathfrak{B}^G$.

We also introduce a multiplication on $\ol{\mathfrak{B}}^G$ defined by
\begin{equation}
(ab)^{(p,q)} \;\deq\; \sum_{\substack{p_1 + p_2 = p \\ q_1 + q_2 = q}} (-1)^{p_2(p_1 + q_1)} \, P_- (a^{(p_1, q_1)} \otimes b^{(p_2, q_2)}) P_-\,.
\end{equation}
The seemingly odd choice of sign will soon become clear. It is now easy to check that $\ol{\mathfrak{B}}^G$ is an associative Banach algebra with identity
\begin{equation*}
\umat^{(p,q)} \;=\; \delta_{p0} \, \delta_{q0}\,.
\end{equation*}
Note that $\ol{\mathfrak{B}}^G$ bears a $\Z$-grading, with degree map 
\begin{equation*}
\deg a^{(p,q)} \;\deq\; p-q\,. 
\end{equation*}
An observable is gauge invariant when its degree is equal to $0$.
One readily sees that
\begin{equation*}
ab = (-1)^{\deg a \, \deg b} ba\,.
\end{equation*}

We now identify $\ol{\mathfrak{B}}^G$ with a Grassmann algebra of anticommuting variables. For $f \in \mathcal{H}$ define $\psi(f) \in \cal L(\mathcal{H}; \C) \subset \ol{\mathfrak{B}}^G$ through
\begin{equation}
\psi(f) \, g \;\deq\; \scalar{f}{g}
\end{equation}
and $\ol{\psi}(f) \in \cal L(\C; \mathcal{H}) \subset \ol{\mathfrak{B}}^G$ through
\begin{equation}
\ol{\psi}(f) z \;\deq\; f z\,.
\end{equation}
We may now consider arbitrary polynomials in the variables $\{\ol{\psi}(f), \psi(f) \,:\, f \in \mathcal{H}\}$. It is a simple matter to check that
\begin{equation*}
\psi(f) \, \psi(g) + \psi(g) \, \psi(f) \;=\; \psi(f) \ol{\psi}(g) + \ol{\psi}(g) \, \psi(f) \;=\; \ol{\psi}(f) \, \ol{\psi}(g) + \ol{\psi}(g) \, \ol{\psi}(f) \;=\; 0\,, 
\end{equation*}
for all $f,g \in \mathcal{H}$. Furthermore, we have that
\begin{equation} \label{polynomial in psis}
\ol{\psi}(f_p) \cdots \ol{\psi}(f_1) \psi(g_1) \cdots \psi(g_q) \;=\; P_-^{(p)} \, \ket{f_1 \otimes \cdots \otimes f_p} \bra{g_1 \otimes \cdots \otimes g_q} \, P_-^{(q)}\,. 
\end{equation}
Linear combinations of expressions of the form \eqref{polynomial in psis} are dense in $\ol{\mathfrak{B}}^G$ (in the strong operator topology).
It is often convenient to write a family $a$ of bounded operators using the ``Grassmann generators'' $\{\ol{\psi}, \psi\}$. To this end we set
\begin{equation*}
\psi(x) \;\deq\; \psi(\delta_x) \,,\qquad \ol{\psi}(x) \;\deq\; \ol{\psi}(\delta_x)\,,
\end{equation*}
where $\delta_x$ is Dirac's delta mass at $x$.
Expressions of the form \eqref{polynomial in psis} are now understood as densely defined quadratic forms. One immediately finds
\begin{multline}
a \;\equiv\; \A^{\! G}(a) \;\deq\; \sum_{p,q} \int \dd x_1 \dots \dd x_p \, \dd y_1 \dots \dd y_q
\\
{}\times{} \ol{\psi}(x_p) \cdots \ol{\psi}(x_1) \, a^{(p,q)}(x_1, \dots, x_p; y_1, \dots, y_q) \, \psi(y_1) \cdots \psi(y_q)\,.
\end{multline}
We use the notation $\A^{\! G}(a)$ to emphasize that the family $a$ is represented using Grassmann generators.

\subsection{A graded Poisson bracket}
Next, we note that $\mathfrak{B}^G$ carries the graded Poisson bracket
\begin{equation}
\{a, b\} \;\deq\; \ii \int \dd x \; \qbb{a \, \frac{\overset{\leftarrow}{\delta}}{\delta \ol{\psi}(x)}
\frac{\overset{\rightarrow}{\delta}}{\delta \psi(x)} \, b
+
a \, \frac{\overset{\leftarrow}{\delta}}{\delta \psi(x)}
\frac{\overset{\rightarrow}{\delta}}{\delta \ol{\psi}(x)} \, b}\,,
\end{equation}
where $a,b \in \mathfrak{B}^G$. Here we use the usual conventions for derivatives with respect to Grassmann variables (see e.g.\ \cite{Salmhofer}, Appendix B).
In terms of kernels the graded Poisson bracket can be expressed as
\begin{equation}
\{\psi(x), \ol{\psi}(y)\} \;=\; \ii \delta(x-y) \, \qquad \{\psi(x), \psi(y)\} \;=\; \{\ol{\psi}(x), \ol{\psi}(y)\} \;=\; 0\,.
\end{equation}
We now list the important properties of the graded Poisson bracket.
\begin{itemize}
\item[(i)]
$\{a, b\} \;=\; (-1)^{1 + \deg a \, \deg b} \{b,a\}$\,.
\item[(ii)]
$(-1)^{\deg b \, (\deg a +  \deg c)} \{a,\{b, c\}\} + \text{cyclic permutations} \;=\; 0$\,.
\item[(iii)]
$\{a, bc\} \;=\; \{a, b\} c + (-1)^{\deg a \, \deg b} b \{a,c\}$\,.
\end{itemize}
\begin{proof}
Let us start with (i):
\begin{align*}
\{a,b\} &\;=\; \ii \int \dd x \, \qBB{(-1)^{\deg a + 1}\frac{\delta a}{\delta \ol{\psi}(x)} \frac{\delta b}{\delta \psi(x)} + (-1)^{\deg a + 1} \frac{\delta a}{\delta \psi(x)} \frac{\delta b}{\delta \ol{\psi}(x)}}
\\
&\;=\; \ii \int \dd x \, \qBB{(-1)^{\deg a \, \deg b + \deg b} \frac{\delta b}{\delta \psi(x)} \frac{\delta a}{\delta \ol{\psi}(x)} + (-1)^{\deg a \, \deg b + \deg b} \frac{\delta b}{\delta \ol{\psi}(x)} \frac{\delta a}{\delta \psi(x)}}
\\
&\;=\; (-1)^{1 + \deg a\, \deg b} \{b,a\}\,.
\end{align*}
In order to show (ii), we note that the left-hand side can be written as a sum of three terms, the first of which contains second derivatives of $a$, the second second derivatives of $b$ and the third second derivatives of $c$. Let us consider the third one. It is equal to the terms containing second derivatives of $c$ of
\begin{multline*}
(-1)^{\deg b (\deg a + \deg c)} \{a, \{b,c\}\} + (-1)^{\deg c (\deg b + \deg a)} \{b, \{c,a\}\}
\\
=\; (-1)^{\deg b (\deg a + \deg c)} \{a, \{b,c\}\} + (-1)^{\deg c (\deg b + \deg a) + 1 + \deg c \, \deg a} \{b, \{a,c\}\}\,,
\end{multline*}
where (i) was used. Define the derivation $L_a b \deq \{a,b\}$\,. Thus we need to compute the terms containing second derivatives of $c$ of
\begin{equation*}
(-1)^{\deg a \, \deg b + \deg b \, \deg c} L_a L_b c - (-1)^{\deg b \, \deg c} L_b L_a c\,.
\end{equation*}
Since we are only considering terms containing second derivatives of $c$, both derivations $L_a$ and $L_b$ must act only on $c$, and one finds
\begin{equation*}
(-1)^{\deg a \, \deg b + \deg b \, \deg c} L_a L_b c - (-1)^{\deg a \, \deg b + \deg b \, \deg c} L_a L_b c \;=\; 0\,.
\end{equation*}
We omit the straightforward proof of (iii).
\end{proof}
Furthermore, one finds by explicit calculation
\begin{multline} \label{graded Poisson bracket computed}
\{\A^{\! G}(a^{(p_1,q_1)}), \A^{\! G}(b^{(p_2,q_2)})\} \;=\; \ii \, (-1)^{(p_2 + 1)(p_1 + q_1)}  q_1 p_2 \, \A^{\! G} \qB{\pb{a^{(p_1, q_1)} \otimes \umat^{(p_2 - 1)}} \, \pb{\umat^{(q_1 - 1)} \otimes b^{(p_2, q_2)}}} 
\\
{}-{}\ii \, (-1)^{(q_1 + 1)(p_2 + q_2)} p_1 q_2 \, \A^{\! G} \qB{\pb{b^{(p_2, q_2)} \otimes \umat^{(p_1 - 1)}} \, \pb{\umat^{(q_2 - 1)} \otimes a^{(p_1, q_1)}}}\,.
\end{multline}

\subsection{States}
We now introduce a space of states $\mathfrak{R} \subset \pb{\ol{\mathfrak{B}}^G}^*$ on the algebra $(\ol{\mathfrak{B}}^G, \norm{\cdot}_{\mathfrak{B}^G})$. A convenient choice is
\begin{equation*}
\mathfrak{R} \;\deq\; \hb{\rho = (\rho^{(p,q)})_{p,q \in \N} \,:\, \rho^{(p,q)} \in \cal L(\mathcal{H}_-^{(q)}; \mathcal{H}_-^{(p)})\,, \; \norm{\rho}_{\mathfrak{R}} < \infty}\,,
\end{equation*}
where
\begin{equation*}
\norm{\rho}_{\mathfrak{R}} \;\deq\; \sup_{p,q \in \N} \norm{\rho^{(p,q)}}_1
\end{equation*}
and
\begin{equation*}
\norm{\rho^{(p,q)}}_1 \;\deq\; \sup \hb{\absb{\tr (\rho^{(p,q)} a^{(q,p)})} \,:\, a^{(q,p)} \in \cal L(\mathcal{H}_-^{(p)}, \mathcal{H}_-^{(q)}) \,,\; \norm{a^{(q,p)}} \leq 1}\,.
\end{equation*}
Note that if $p = q$ then $\norm{\cdot}_1$ is the usual trace norm. The dual action is given by
\begin{equation*}
\scalar{\rho}{a} \;\deq\; \sum_{p,q \in \N} \tr (\rho^{(p,q)} a^{(q,p)})\,.
\end{equation*}
We abbreviate $\rho^{(p,p)} \equiv \rho^{(p)}$ in the case of gauge invariant states. Next, we note that \eqref{polynomial in psis} implies that the operator kernel of $\rho^{(p,q)}$ may be expressed as
\begin{equation} \label{operator kernel of rho}
\rho^{(p,q)}(x_1, \dots, x_p;y_1, \dots, y_q) \;=\; \scalarb{\rho}{\ol{\psi}(y_q) \cdots \ol{\psi}(y_1) \psi(x_1) \cdots \psi(x_p)}\,.
\end{equation}

There is a particular subset of gauge invariant states that is of interest for studying the Hartree-Fock dynamics. Let $\gamma \in \mathcal{D}$ be a one-particle density matrix. Define the state $\rho_\gamma$ through $\rho_\gamma^{(p,q)} = 0$ if $p \neq q$ and
\begin{equation} \label{general state from density matrix}
\rho_\gamma^{(p,p)} \;\deq\; \gamma^{(p)}\,,
\end{equation}
where $\gamma^{(p)}$ is defined in \eqref{p-particle action of quasi-free state}. One immediately finds $\norm{\rho_\gamma}_1 = \norm{\gamma}_1$. 

\subsection{Hamilton function and dynamics}
Let $h$ be the one-particle Hamiltonian and $w$ the two-body interaction potential. We define a Hamilton function on (a dense 
subset of) the phase space $\mathfrak{R}$ through
\begin{align}
H &\;\deq\; \A^{\! G}(h) + \frac{1}{2} \A^{\! G}(W) 
\notag \\ \label{Grassmann Hamiltonian}
&\;=\; \int \dd x \, \dd y \; \ol{\psi}(x) \, h(x;y) \, \psi(y) + \frac{1}{2} \int \dd x \, \dd y \; \ol{\psi}(y) \ol{\psi}(x) \, w(w - y) \, \psi(x) \psi(y)\,.
\end{align}

The Hamiltonian equations of motion read
\begin{equation*}
\dot{a} \;=\; \{H, a\}\,,
\end{equation*}
where $a \in \mathfrak{B}^G$.
Instead of the ``Heisenberg'' evolution of $a$ we consider the dual ``Schr\"odinger'' evolution of states:
\begin{equation*}
\scalar{\rho(t)}{a} \;\deq\; \scalar{\rho}{a(t)}\,.
\end{equation*}
The equation of motion for states reads
\begin{multline} \label{equation of motion for grassmann states}
\ii \partial_t \rho^{(p,q)}(x_1, \dots, x_p; y_1, \dots, y_q) \;=\; \pBB{\sum_{i = 1}^p h_{x_i} - \sum_{i = 1}^q h_{y_i}} \rho^{(p,q)}(x_1, \dots, x_p; y_1, \dots, y_q)
\\
+ \int \dd u \; \pBB{\sum_{i = 1}^p w(u - x_i) - \sum_{i = 1}^q w(u - y_i)} \rho^{(p+1,q+1)}(x_1, \dots, x_p, u; y_1, \dots, y_q, u)\,.
\end{multline}
This has the form of an infinite hierarchy, which decouples over subspaces of different degree.
In order to show \eqref{equation of motion for grassmann states} we compute 
\begin{multline*}
\ii \hb{H, \ol{\psi}(y_q) \cdots \ol{\psi}(y_1) \psi(x_1) \cdots \psi(x_p)} \;=\; \pBB{\sum_{i = 1}^p h_{x_i} - \sum_{i = 1}^q h_{y_i}} \, \ol{\psi}(y_q) \cdots \ol{\psi}(y_1) \psi(x_1) \cdots \psi(x_p)
\\
{}+{} \int \dd u \pBB{\sum_{i = 1}^p w(u - x_i) - \sum_{i = 1}^q w(u - y_i)}
\ol{\psi}(u) \ol{\psi}(y_q) \cdots \ol{\psi}(y_1) \psi(x_1) \cdots \psi(x_p) \psi(u)\,.
\end{multline*}
Then \eqref{equation of motion for grassmann states} follows from \eqref{operator kernel of rho} and
\begin{align*}
\ii \partial_t \rho^{(p,q)}(x_1, \dots, x_p;y_1, \dots, y_q) &\;=\; \ii \partial_t \scalarb{\rho}{\ol{\psi}(y_q) \cdots \ol{\psi}(y_1) \psi(x_1) \cdots \psi(x_p)}
\\
&\;=\; \scalarb{\rho}{\ii \hb{H, \ol{\psi}(y_q) \cdots \ol{\psi}(y_1) \psi(x_1) \cdots \psi(x_p)}}\,.
\end{align*}

Next, we outline how to solve the equation of motion \eqref{equation of motion for grassmann states}. Let us first rewrite it as
\begin{multline*}
\ii \partial_t \rho^{(p,q)} \;=\; \sum_{i = 1}^p h_i \rho^{(p,q)} - \sum_{i = 1}^q \rho^{(p,q)} h_i
\\
{}+{} \sum_{i = 1}^p \tr_{p+1, q+1} \pb{W_{i \; p+1} \rho^{(p+1, q+1)}} - \sum_{i = 1}^q \tr_{p+1, q+1} \pb{\rho^{(p+1, q+1)} W_{i \, q+1}}\,.
\end{multline*}
We may now proceed exactly as with the density matrix Hartree-Fock equation \eqref{Hartree-Fock for density matrices}, i.e.\ express it as an integral equation in the interaction picture. This yields a tree expansion for the quantity $\tr\pb{\rho^{(p,q)}(t) \, a^{(q,p)}}$, where $\rho(0) \in \mathfrak{R}$. We omit the uninteresting details. As above, the tree expansion converges if $t < T$, where
\begin{equation} \label{definition of T}
T \;\deq\; (2^{11} \pi \kappa^2)^{-1}\,.
\end{equation}
Unfortunately, the time evolution \eqref{equation of motion for grassmann states} does not preserve the norm of $\rho$, which means that we cannot iterate the short-time result.

\subsection{A Schwinger-Dyson series for the time evolution of observables}
From now on, we only consider gauge invariant quantities. Take some gauge invariant state $\rho 
= (\rho^{(p)})_{p \in \N} \in \mathfrak{R}$. For simplicity, we assume that the sequence $\rho$ 
is finite (as is the case if $\rho$ is defined by a Slater determinant, see below). Let us 
denote the Hamiltonian flow on $\mathfrak{R}$ by $\phi^t$. We have seen that $\phi^t$ is 
well-defined by its tree expansion for $t < T$. The solution of \eqref{equation of motion for 
grassmann states} with initial data $\rho$, $\rho(t) = \phi^t(\rho)$, satisfies the equation
\begin{equation} \label{grassman state duhamel}
\tilde{\rho}^{(p)}(t) \;=\; \rho^{(p)} - \ii \int_0^t \dd s \; \sum_{i = 1}^p\tr_{p+1} \qb{W_{i\;p+1,s}, \tilde{\rho}^{(p+1)}(s)}\,,
\end{equation}
where $\tilde{\rho}^{(p)}(t) \deq \ee^{\ii \sum_i h_i t} \, \rho^{(p)}(t) \, \ee^{-\ii \sum_i h_i t}$. Let us take a gauge invariant observable $a^{(p,p)} \equiv a^{(p)} \in \mathfrak{A}^G$, where
\begin{equation*}
\mathfrak{A}^G \;\deq\; \{a \in \mathfrak{B}^G \,:\, a^{(p,q)} = 0 \text{ if } p \neq q\}
\end{equation*}
is the set of gauge invariant observables. Then \eqref{grassman state duhamel} implies
\begin{equation*}
\tr \pb{a^{(p)} \rho^{(p)}(t)} \;=\; \tr \pb{a^{(p)}_t \tilde{\rho}^{(p)}(t)} 
\;=\; \tr \pb{a^{(p)}_t \rho^{(p)}} + \ii \int_0^t \dd s \; \sum_{i = 1}^p \tr \pB{\qB{W_{i\;p+1, s}, a^{(p)}_t \otimes \umat} \tilde{\rho}^{(p+1)}(s)}\,.
\end{equation*}
Iteration of this identity gives
\begin{equation*}
\tr \pb{a^{(p)} \rho^{(p)}(t)} \;=\; \sum_{k = 0}^\infty \tr \pB{G_t^{(k,0)}(a^{(p)}) \, \rho^{(p+k)}}\,.
\end{equation*}
Summarizing:
\begin{multline*}
\scalarb{a^{(p)} \circ \phi^t}{\rho} \;=\; \scalarb{a^{(p)}}{\rho(t)} \;=\; \tr \pb{a^{(p)} \rho^{(p)}(t)} 
\\
=\; \sum_{k = 0}^\infty \tr \pB{G_t^{(k,0)}(a^{(p)}) \, \rho^{(p+k)}} \;=\; \scalarBB{\sum_{k = 0}^\infty G_t^{(k,0)}(a^{(p)})}{\rho}\,.
\end{multline*}
This series converges for $t < T$, uniformly for bounded $\norm{a^{(p)}}_{\mathfrak{B}^G}$ and $\norm{\rho}_{\mathfrak{R}}$. Therefore we get the norm-convergent series
\begin{equation}  \label{tree expansion for superhamiltonian flow}
\A^{\! G}(a^{(p)}) \circ \phi^t \;=\; \sum_{k = 0}^\infty \A^{\! G} \pB{G_t^{(k,0)}(a^{(p)})}\,,
\end{equation}
provided that $t < T$.

Finally, we discuss the relationship between the Hartree-Fock dynamics and the dynamics generated by \eqref{equation of motion for grassmann states}. Take a density matrix $\gamma \in \mathcal{D}$ and consider the state $\rho = \rho_\gamma$ defined in \eqref{general state from density matrix}. If one chooses a sequence $\gamma_N$ such that $\norm{\gamma_N} \to 0$ as $N \to \infty$ (e.g.\ a sequence of Slater determinants), then Lemma \ref{lemma HF and tree terms} implies that \eqref{equation of motion for grassmann states} and the Hartree-Fock equation describe the same dynamics for large $N$.

\subsection{Quantization and a Egorov-type theorem}
In this final section we introduce a Wick quantization of the above ``superhamiltonian'' system and formulate the mean-field limit as a Egorov-type theorem. It is advantageous to use the second quantized formulation of Section \ref{section: second quantization}.

We define \emph{quantization}, denoted by $\widehat{(\cdot)}_\nu$, as the linear map defined by the replacement $\psi(x) \mapsto a_\nu(x)$ and $\ol{\psi}(x) \mapsto a^*_\nu(x)$, followed by Wick ordering. Quantization $\widehat{(\cdot)}_\nu$ maps observables in $\fra B^G$ to closed operators on $\cal F$. Moreover, we have
\begin{equation*}
\widehat{(\cdot)}_\nu \,:\, \A^{\! G}(a^{(p)}) \;\mapsto\; \widehat{\A}_\nu(a^{(p)})\,.
\end{equation*}
Using \eqref{graded Poisson bracket computed}, it is easy to see that, for $A,B \in \mathfrak{A}^G$,
\begin{equation*}
\qb{\widehat{A}_\nu, \widehat{B}_\nu} \;=\; \frac{\nu^{-1}}{\ii} \widehat{\h{A, B}}_\nu + O(\nu^{-2})\,.
\end{equation*}
This identifies $\nu^{-1}$ as the deformation parameter of $\widehat{(\cdot)}_\nu$.

Extending the definition of $\widehat{(\cdot)}_\nu$ to unbounded operators in the obvious way, we define a Hamiltonian $\widehat{H}_\nu$ on $\cal F$ as the quantization of the Hamilton function $H$ defined in \eqref{Grassmann Hamiltonian}. When restricted to $\mathcal{H}^{(N)}_-$, $N \widehat{H}_N$ is equal to the Hamiltonian with mean-field scaling \eqref{mean-field Hamiltonian}.

Now \eqref{one loop expansion}, \eqref{estimate on loop terms} and \eqref{tree expansion for superhamiltonian flow} yield the following Egorov-type theorem.
\begin{theorem}
Let $A \in \mathfrak{A}^G$ and $t < T$, with $T$ defined in \eqref{definition of T}. Then
\begin{equation*}
\normbb{\pbb{\ee^{\ii t N \widehat{H}_N} \, \widehat{A}_N\, \ee^{-\ii t N \widehat{H}_N}  - \widehat{(A \circ \phi^t)}_N}\biggr|_{\mathcal{H}^{(N)}_-}} \;\leq\; \frac{C}{N}\,,
\end{equation*}
for some $C > 0$.
\end{theorem}

\appendix
\section{Hamiltonian formulation for density matrices}
In Section \ref{section: Hartree-Fock equation} we chose a Hamiltonian formulation of the Hartree-Fock equation \eqref{Hartree-Fock} in terms of sequences of orbitals. Alternatively, we could just as well have used a Hamiltonian formulation in terms of density matrices. To see how the density matrix Hartree-Fock equation \eqref{Hartree-Fock for density matrices} can be written as a Hamiltonian equation of motion of a classical Hamiltonian system, consider the Hilbert space
\begin{equation*}
\widehat{\mathcal{H}} \;=\; \mathcal{L}^2(\mathcal{H})\,,
\end{equation*}
the space of Hilbert-Schmidt operators, with scalar product
\begin{equation*}
\scalar{\kappa}{\rho} \;\deq\; \tr (\kappa^* \rho)\,.
\end{equation*}
We write the density matrix $\gamma \in \mathcal{D}$ as $\gamma = \kappa \kappa^*$, where $\kappa \in \widehat{\mathcal{H}}$.
The classical phase space is then given by a Sobolev-type space of Hilbert-Schmidt operators
\begin{equation*}
\widehat{\Gamma} \;\deq\; \hb{\kappa \in \widehat{\mathcal{H}} \,:\, \tr (\kappa^* (\umat - \Delta) \kappa) < \infty}\,.
\end{equation*}
We define polynomial functions on $\widehat{\Gamma}$ through
\begin{equation*}
\B(a^{(p)})(\kappa) \;\deq\; \scalarb{\kappa^{\otimes p}}{a^{(p)} \kappa^{\otimes p}}\,,
\end{equation*}
where $a^{(p)} \in \cal L(\mathcal{H}^{\otimes p})$. 

The affine space $\widehat{\Gamma}$ carries a Symplectic form defined by
\begin{equation*}
\omega \;=\; -\ii \int \dd x \, \dd y \; \dd \bar{\kappa}(x,y) \wedge \dd \kappa(x,y)\,,
\end{equation*}
where $\kappa(x,y)$ is the operator kernel of $\kappa$. The Poisson bracket then reads
\begin{align*}
\h{\kappa^\#(x,y), \kappa^\#(x',y')} &\;=\; 0
\\
\h{\kappa(x,y), \bar{\kappa}(x',y')} &\;=\; - \ii \delta(x - x') \delta(y - y')\,.
\end{align*}
The Hamilton function is defined by
\begin{equation*}
H \;\deq\; \B(h) + \frac{1}{2} \B(\mathcal{W})\,.
\end{equation*}
By using Sobolev-type inequalities one readily sees that $H$ is well-defined on $\widehat{\Gamma}$.
After a short computation, one finds that the Hamiltonian equation of motion,
\begin{equation*}
\ii \partial_t \kappa(x,y) \;=\; \frac{\delta H}{\delta \ol{\kappa}(x,y)} \;=\; \ii \, \{H, \kappa(x,y)\}\,,
\end{equation*}
reads
\begin{equation*}
\ii \partial_t \kappa \;=\; h \kappa + \tr_2 \pb{\mathcal{W} \, \kappa \otimes (\kappa \kappa^*)}\,.
\end{equation*}
It follows that $\gamma = \kappa \kappa^*$ satisfies \eqref{Hartree-Fock for density matrices}.

\end{document}